\documentclass[12pt,twoside,english]{article}
\PassOptionsToPackage{natbib=true, maxnames=99}{biblatex}
\usepackage[utf8]{inputenc}
\usepackage{geometry}
\usepackage{verbatim}
\usepackage{url}
\usepackage{amsmath}
\usepackage{amsthm}
\usepackage{amssymb}
\usepackage{setspace}
\usepackage[hidelinks]{hyperref}

\usepackage{todonotes}

\makeatletter

\DeclareTextSymbolDefault{\textquotedbl}{T1}

\theoremstyle{plain}
\newtheorem{assumption}{\protect\assumptionname}
\theoremstyle{definition}
\newtheorem{defn}{\protect\definitionname}
\theoremstyle{plain}
\newtheorem{lem}{\protect\lemmaname}
\theoremstyle{definition}
 \newtheorem{example}{\protect\examplename}
\theoremstyle{plain}
\newtheorem{thm}{\protect\theoremname}
\theoremstyle{remark}
\newtheorem*{rem*}{\protect\remarkname}

\usepackage[font={footnotesize,it}]{caption}

\makeatother

\usepackage{babel}
\usepackage[style=ext-authoryear,articlein=false,]{biblatex}
\providecommand{\assumptionname}{Assumption}
\providecommand{\definitionname}{Definition}
\providecommand{\examplename}{Example}
\providecommand{\lemmaname}{Lemma}
\providecommand{\remarkname}{Remark}
\providecommand{\theoremname}{Theorem}

\begin{document}
\title{Stable Matching with Peer-Dependent Preferences:
Existence and Cutoff Characterization}
\author{Jacob D. Leshno\thanks{{\footnotesize{}University of Chicago Booth School of Business, \protect\url{jacob.leshno@ChicagoBooth.edu}.
I thank Claudia Allende, Eduardo Azevedo, Guilherme Carmona, Chris Campos, Yeon-Koo Che, John Hatfield, Adam Kapor, Krittanai Laohakunakorn, Jack Mountjoy, Bobby Pakzad-Hurson, and Alvin Roth for helpful comments and discussions.
I thank Anand Shah for research support. This work is supported by the Robert H. Topel Faculty Research Fund at the University of Chicago Booth School of Business.}}}
\maketitle
\begin{abstract}
This paper presents a framework for stable matching markets that accommodates the peer-dependent preferences employed in empirical work. We show the existence of a stable matching in a continuum economy and of an approximately stable matching in a large finite sampled economy under assumptions satisfied by standard empirical specifications. The analysis builds on a tractable characterization of stable matchings with peer-dependent preferences in terms of supply and demand equations along with a rational-expectations condition. 
\end{abstract}

\newpage{}

\section{Introduction }

\begin{verse}
\noindent {\small{}``There is a collective action issue $\dots$ there's no French teacher for one student. But there's a program if 15 come, if 20 come. But we all have to, then, take one step forward at exactly the same time.'' \vspace{-1cm} }{\small\par}
\end{verse}
\noindent \begin{flushright}
{\small{}-- Nice White Parents}\footnote{Excerpt from episode 1 of the podcast Nice White Parents by Serial and The New York Times \citep{joffe-walt2020introducing}. }
\par\end{flushright}

Motivated by evidence that peer composition affects students' preferences over schools \citep{abdulkadirouglu2020parents,rothstein2006good,beuermann2023good}, empirical work on school choice increasingly incorporates peer-dependent preferences \citep{epple2018superintendent,allende2019competition,campos2024social}.
Similar peer-dependent preferences arise in urban models of location sorting with endogenous amenities \citep{bayer2007unified,diamond2016determinants,almagro2025location}, as well as in models of worker--firm sorting with endogenous firm productivity \citep{abowd1999high,bender2018management}.

Matching theory has played a central role in the design and empirical analysis of school choice programs \citep{abdulkadirouglu2003school,abdulkadirouglu2005new,agarwal2020revealed}, and a substantial theoretical literature has studied matching in the presence of peer effects and other externalities.\footnote{See the literature review below for a detailed discussion.} 
However, a central difficulty is that a stable matching may fail to exist under peer-dependent preferences. Such nonexistence may arise from two different sources. 
In some environments, nonexistence reflects a fundamental economic challenge that persists even in large markets. In others, nonexistence is only due to indivisibilities, and approximate solutions exist in sufficiently large markets. This distinction matters for empirical work that requires an equilibrium concept to deliver predictions in counterfactual environments. If nonexistence is driven only by such ``rounding errors,'' then approximate stability may be an empirically relevant solution concept.

We extend the continuum supply-and-demand framework of \citet{azevedo2016supply} to study stable matchings with peer-dependent preferences. This approach has three advantages. First, it allows us to abstract from indivisibilities, and therefore distinguish failures of stability that reflect fundamental economic forces from those that arise only from finite-market ``rounding errors'' caused by indivisibilities. Second, it retains the capacity constraints, priorities, and peer-dependent preferences that are central to empirical applications. Third, the supply-and-demand representation describes each student's choice as demand for her most preferred school among those for which she meets the admission cutoff. This formulation maps naturally into the discrete-choice models commonly used in empirical analyses of school choice.

In the model, student preferences over schools may flexibly depend on \emph{matching statistics}, which summarize the characteristics of students assigned to each school. These statistics can capture the types of peer characteristics commonly used in empirical work. For example, a student's preferences may depend on the average GPA of students assigned to a school, the gender composition of the school or neighboring schools, or whether sufficiently many students enroll to support a specialized language program. Preferences can be arbitrary functions of these matching statistics, allowing peer composition to affect preferences in general and even discontinuous ways.
The implied restriction is that preferences depend on peers through summary characteristics of the assignment rather than through the identities of particular students.\footnote{In particular, the formulation excludes settings such as couples seeking joint assignments, where preferences depend directly on the assignment of a particular individual; see, e.g., \citet{klaus2005stable}.} 

The standard notion of stability extends naturally to this environment: a matching is stable if there is no student--school pair such that the school is willing to admit the student and, given the matching statistics generated by the assignment, the student prefers that school to her assigned school.

We extend the \citet{azevedo2016supply} cutoff characterization and show that stable matchings are equivalent to cutoff--statistics pairs that satisfy market-clearing and rational-expectations conditions. This characterization represents stable matchings as solutions to a system of demand equations. The cutoffs and market-clearing condition capture the standard pairwise-stability logic, while the matching statistics and rational-expectations condition capture the coordination problem created by peer-dependent preferences.

Cutoffs are admission thresholds for each school. Analogous to prices, they determine each student's budget set: a student can attend a school if her priority exceeds that school's cutoff. Given a cutoff--statistics pair, the student's demand is her most preferred school among those in her budget set. Student preferences are those induced by the conjectured matching statistics. A cutoff--statistics pair is market-clearing and rational-expectations if assigning every student to her demanded school does not violate any capacity constraint and the conjectured matching statistics coincide with those generated by the resulting assignment. Every stable matching corresponds to such a cutoff--statistics pair, and conversely every market-clearing and rational-expectations cutoff--statistics pair induces a stable matching.

The demand characterization brings the model closer to a standard competitive-equilibrium framework and allows us to use the fixed-point approach commonly employed in that literature. Existence is not automatic, however. Standard fixed-point arguments require continuity, and demand need not be continuous in the matching statistics, even in a continuum economy. We provide a simple example showing that this is not merely a technical obstacle: a stable matching may fail to exist even in the continuum.

Despite this nonexistence example, stable matchings exist under a broad and natural condition.
Aggregate demand can be discontinuous only if an arbitrarily small change in matching statistics causes a positive mass of students to simultaneously change their ordinal rankings.
We rule out such (pathological) cases by introducing an assumption we term 
\emph{diversity of preferences}. This assumption requires that small changes in 
matching statistics affect only a small fraction of students' ordinal rankings. Diversity of preferences thereby ensures continuity of aggregate 
demand\footnote{Notably, our existence result does not require individual preferences to be continuous, nor does it rely on students being indifferent between schools.} and, hence, the existence of a stable matching. We argue that the 
assumption is natural and show in Section \ref{sec:App-econ-w-taste-shocks} that it is satisfied by standard empirical specifications.

The continuum existence result suggests that nonexistence in large finite markets may stem from indivisibilities, and that an approximately stable matching should exist. We formalize this intuition by considering finite economies sampled from a continuum economy. Under mild additional local regularity conditions, sufficiently large sampled economies admit an approximately stable matching with high probability. Our notion of approximate stability preserves capacity constraints exactly and requires only that any blocking pair offer the student at most a small gain. Thus, arbitrarily small frictions may suffice to sustain such a matching in practice, without requiring flexible capacities or other adjustments by colleges or the market designer.

The large-market approach, however, does not restore the full familiar structure of stable matching. As previous work has shown, peer-dependent preferences can lead to multiple stable matchings and failures of the lattice structure and the Rural Hospital Theorem. These failures persist even in the continuum economy: they reflect the coordination problem created by peer-dependent preferences, rather than finite-market indivisibilities. We provide examples illustrating these failures and the economic forces behind them.

\subsubsection*{Related literature}

This paper is motivated by empirical findings showing that student preferences for colleges or public schools may depend on the assignment of peers. Such dependency can arise for multiple reasons: Composition of peers affects student achievement (\citealt{duflo2011peer,lavy2011mechanisms,lavy2012inside,booij2017ability,schiltz2019does}); Labor market returns to a college major may depend on the number of students enrolling in related majors (\citealt{bianchi2020indirect}); Parents\textquoteright{} preferences for a school can depend on the composition of its enrolled students (\citealt{epple1998competition,epple2018superintendent,abdulkadirouglu2020parents,rothstein2006good,allende2019competition}); Additionally, students with incomplete information may update their preferences when observing the choices of others (\citealt{chakraborty2010two,liu2014stable}).

Previous work has been insightful and successful in showing how stable matching theory breaks with peer-dependent preferences, but gave positive results under stronger assumptions than what empirical model need. 

Prior research has been very informative in showing how classical stable-matching results can fail when preferences depend on peer assignments, and has established positive results under economically meaningful restrictions. However, these positive results impose restrictions that are too strong to accommodate the preference specifications commonly used in empirical work.
\citet{sasaki1996two} shows that stable matching may not exist in a finite indivisible matching market. \citet{echenique2007solution} considers a discrete many-to-one matching with peer-dependent preferences and develops an algorithm that finds core matchings if they exist, and suggests partial solutions for when the core is empty. 
The approach in this paper is closely related to that of \citet{fisher2016matching}, who study one-to-one matching with externalities in a large-market framework in which the externality generated by any individual agent is bounded. Their model, however, does not exploit the tractability afforded by a cutoff characterization and restricts peer effects to one-dimensional aggregate externalities. 
\citet{calsamiglia2020school} analyzes the Boston mechanism and allows student preferences to depend on a quality measure. \citet{dur2019school} considers students with neighbor-dependent preferences. \citet{che2022prestige} studies prestige seeking in competitive college admission. \citet{phan2024crowding} studies school assignment where students have preferences over crowding and suggest a new equilibrium notion that accommodates crowding. 

\citet{pycia2012stability} shows the existence of a stable matching when agent preferences are aligned. \citet{mumcu2010stable} provides a sufficient condition on the whole preference profile that guarantees the existence of a stable matching. Along with two conditions to restrict externalities in firms' choice sets, \citet{bando2012many} shows that stable matchings exist under an extension of substitutability. \citet{baccara2012field} analyzes an assignment problem with externalities due to benefits from proximity to others in a network. \citet{pycia2019matching} shows that a stable matching is guaranteed to exist when peer-dependent preferences satisfy an elegant substitutes assumption, but may fail to exist otherwise. \citet{rostek2020matching} shows existence under matching with complementary contracts. 

The previous literature imposes restrictions on student preferences that may not be compatible with empirical needs. Consider, for example, that one STEM-oriented student may seek to attend a school with other strong STEM students to benefit from having strong peers, while another student may seek to avoid the increased competition and pressure. A model with both types of students will violate the restriction on preferences required by the literature above. In contrast to these papers, the present paper allows for general peer-dependent preferences. Because we allow for preferences that violate the substituability conditions identified by \cite{pycia2019matching}, their results imply that a stable matching may fail to exist in a finite market. We circumvent the nonexistence of a stable matching by considering an approximate stable matching.

Contemporaneous work of \citet{Cox2022peer} studies a related question in economies where peer effects arise because student performance depends on their ranking within their class. Using data on college applicants in New South Wales, they identify peer-dependent preferences, show that standard matching mechanisms may fail to produce stable outcomes under such preferences, and propose a mechanism that restores stability. Subsequent work by \citet{carmona2023existence} shows the existence of a stable matching in a continuum economy when student utilities are continuous. 

\citet{epple1998competition}, \citet{epple2018superintendent}, \citet{barseghyan2019peer} theoretically study school choice in the presence of peer-dependent preferences. \citet{abdulkadirouglu2020parents} develops a framework for estimating student preferences for peer quality and finds that peer quality affects student preferences. \citet{allende2019competition} empirically studies school choice using administrative data from Peru. She estimates a structural model that accounts for peer-dependent preferences and evaluates counterfactual designs. \citet{campos2022impact} evaluates the effect of competition among schools and finds that school choice can improve the quality of schools. These models assume that schools do not face capacity constraints. The model in this paper provides a generalized framework that incorporates capacity constraints and provides a connection between these models and stable matching models. In addition, this paper provides results for finite sampled economies.

The model adds to a growing strand of papers that use matching models with a continuum of agents.\footnote{\cite{ellickson1999clubs} study a continuum model to analyze club formation and show the existence of a competitive equilibrium.} \citet{abdulkadirouglu2015expanding} and later work by \citet{ashlagi2014improving} uses a continuum model to guide the design of school choice mechanisms. The present model builds upon the continuum matching model of \citet{azevedo2016supply} and extends it to allow for peer-dependent preferences. \citet{che2019stable} and \citet{azevedo2018existence} use continuum models to analyze matching markets in which firm preferences exhibit complementarities. \citet{greinecker2018pairwise} and \citet{jagadeesan2017complementary} develop models with a continuum of agents and prove the existence of a stable matching in their models. The present model differs from the approach taken in these papers in that the model does not require continuity or other restrictions on individual preferences to show the existence of a stable matching. Our approach is similar to \citet{azevedo2013walrasian} in that we allow individual preferences to be discontinuous, but our diversity of preferences assumption implies continuity of aggregate demand.

Preferences in our model can depend on the distribution of student attributes assigned to each college, but not on the assignment of an individual student. In that, our model is distinct from matching models in which couples seek joint assignments (\citealt{klaus2005stable,kojima2013matching,ashlagi2014stability}).

\subsubsection*{Organization of the Paper}

Section \ref{sec:Model} presents the model and defines stable matchings for peer-dependent preferences. Section \ref{sec:Cutoff-Char} provides the cutoff characterization. The main result showing the existence of a stable matching and examples showing the necessity of the diversity of preferences assumption are in Section \ref{sec:Existence-and-Structure}. This section also explores the structure of the set of stable matchings and provides examples. Section \ref{sec:Approximately-Stability} defines and shows the existence of approximately stable matching in large sampled economies. Section \ref{sec:App-econ-w-taste-shocks} shows that a standard empirical framework satisfies our assumptions. Section \ref{sec:Conclusion} concludes. Additional examples are in Appendix \ref{sec:Additional-Examples}. Omitted proofs are in Appendix \ref{sec:Omitted-Calculations}.

\section{Matching with Peer-Dependent Preferences\label{sec:Model}}

There are finitely many colleges, and the set of colleges is denoted
$\mathcal{C}=\{c_{1},\dots,c_{J}\}$. Each college $c\in\mathcal{C}$
has capacity to admit a mass of $q_{c}>0$ students. A student is
described by a type $\theta=\left(\chi^{\theta},\gamma^{\theta},r^{\theta}\right)$, where $\chi^{\theta}$ specifies the student's attributes that can enter into other students' peer-dependent preferences, $\gamma^{\theta}$ specifies the student's peer-dependent preferences over colleges, and $r^{\theta}$ captures the student's priority at each college.

The vector $r^{\theta}\in\left[0,1\right]^{\mathcal{C}}$ specifies colleges' priorities for student $\theta$. College $c$ gives higher priority to student $\theta_{1}$ than to student $\theta_{2}$ if $r_{c}^{\theta_{1}}>r_{c}^{\theta_{2}}$. We refer to $r^{\theta}$ as $\theta$'s priority or rank. For brevity of notation, we assume all students and colleges are acceptable.\footnote{That is, any student prefers to match to any college rather than to remain unmatched, and any college that did not fill its capacity is willing to admit any student. 
This assumption simplifies notation without restricting the analysis. A student's outside option can be represented by a dummy college with unlimited capacity, while college-side unacceptability can be represented by a mass of dummy students corresponding to the college's unfilled seats.} 
The attributes of student $\theta$ are given by a vector $\chi^{\theta}\in\mathcal{X}$, where the set $\mathcal{X}\subset\mathbb{R}^{L}$ is bounded, closed, convex, and $0\in\mathcal{X}$. The preferences of student $\theta$ over colleges are parameterized by $\gamma^{\theta}\in\Gamma$, where $\Gamma$ is a subset of a Euclidean space. Student preferences are defined below, and can depend on the entire matching. 
Let  $\Theta=\mathcal{X}\times\Gamma\times\left[0,1\right]^{\mathcal{C}}$ denote the set of all student types.

An economy is $E=[\Theta,\mathcal{C},\eta,q]$ where $\eta$ is a
measure over $\Theta$. It is convenient to normalize $\eta$ so that
$\eta\left(\Theta\right)=1$ and interpret $\eta$ as the joint distribution
of student attributes, preferences, and priorities. Without loss of
generality, we normalize priorities so that $r_{c}^{\theta}$ is $\theta$'s
percentile in college $c$'s ranking.\footnote{Formally, $r_{c}^{\theta}=\eta\left(\left\{ \theta'\left|r_{c}^{\theta}>r_{c}^{\theta'}\right|\right\} \right)$.}

A matching is an $\eta$-measurable function $\mu:\Theta\to\mathcal{C}\cup\left\{ \phi\right\} $,
where we use $\phi$ to denote that a student is unmatched. With slight
abuse of notation we use $\mu\left(c\right)$ instead of $\mu^{-1}\left(c\right)$
to denote the set of students matched to college $c$. To rule out
multiple matchings that differ on a zero measure set, we restrict
attention to right-continuous matchings.\footnote{A matching is right continuous \citep{azevedo2016supply} if for every
$c\in\mathcal{C}$ the set $\left\{ \theta\mid\mu\left(\theta\right)\prec^{\theta|\mu}c\right\} $
is open, with $\prec^{\theta|\mu}$ defined below. } 

Student preferences over colleges can depend on the attributes of
students assigned to each college. Given a matching $\mu$, the \emph{matching
statistics} of college $c$ are

\begin{equation}
\sigma_{c}\left(\mu\mid\eta\right)=\int_{\theta\in\mu\left(c\right)}\chi^{\theta}d\eta\,.\label{eq:match-stats}
\end{equation}
We write $\sigma_{c}\left(\mu\right)$ when $\eta$ is clear from
context, and use $s=\sigma\left(\mu\right)=\bigl(\sigma_{c}\left(\mu\right)\bigr)_{c\in\mathcal{C}}\in\mathcal{X^{C}}\subset\mathbb{R}^{J\times L}$
to denote the full vector of matching statistics.\footnote{The fact that $\sigma_{c}\left(\mu\right)\in\mathcal{X}$ follows
from the definition of $\mathcal{X}$ as a convex and closed set that
contains the origin. The particular $\sigma$ is chosen to simplify
exposition. Our results apply for any $\sigma\left(\cdot\right)$
such that there exists a constant $M>0$ such that $\left\Vert \sigma\left(\mu\right)-\sigma\left(\mu'\right)\right\Vert _{\infty}\leq M\cdot\eta\left(\left\{ \theta\mid\mu\left(\theta\right)\neq\mu'\left(\theta\right)\right\} \right)$
for all $\mu,\mu'$.} We use $s_{c}=\sigma_{c}\left(\mu\right)$ to denote the matching
statistics for college $c$. For example, $s_{c}=\sigma_{c}\left(\mu\right)$
can capture the average SAT score of students assigned to $c$, gender
balance, or the mass of admitted students interested in becoming English
majors.\footnote{Formally, the function $\sigma$ defined above can include both the
total SAT score of students assigned to the college as well as the
total mass of students assigned to the college. We allow student utilities
that are arbitrary functions of the matching statistics and can depend
on the ratio of these statistics, i.e., the average SAT score.} 

We allow student preferences over colleges to depend on the entire matching
$\mu$ through the statistics $s=\sigma\left(\mu\right)$. The utility
of student $\theta=\left(\chi^{\theta},\gamma^{\theta},r^{\theta}\right)$
of being assigned to $c=\mu\left(\theta\right)$ given the matching
$\mu$ is denoted by $u^{\theta}\left(c;\sigma\left(\mu\right)\right)$.
We allow $u^{\theta}\left(c;\sigma\left(\mu\right)\right)$ to be
an arbitrary function of the matching statistics $s=\sigma\left(\mu\right)$,
preference parameters $\gamma^{\theta}$, and the college $c\in\mathcal{C}$.
For brevity, we also use $u^{\theta}\left(c;\mu\right)$ or $u^{\theta}\left(c;s\right)$
to denote the student's utility. Unless stated otherwise, we normalize
the utility of being unmatched to $0$, i.e., $u^{\theta}\left(\phi;\mu\right)=0$
for any student $\theta$ and any matching $\mu$.

For $s=\sigma\left(\mu\right)$ we write $\succ^{\theta|s}$ or $\succ^{\theta|\mu}$
to denote the preference ordering induced by $u^{\theta}\left(c;s\right)$.\footnote{That is, for $s=\sigma\left(\mu\right)$ we write $c_{1}\succ^{\theta|\mu}c_{2}$
if $u^{\theta}\left(c_{1};s\right)>u^{\theta}\left(c_{2};s\right)$.
To simplify notation, we also use $\succ^{\theta|s}$ for $s$ values
which are not in the image of $\sigma\left(\cdot\right)$. We arbitrarily
define $\succ^{\theta|s}$ for such $s$ values.} To simplify notation, we break ties in student preferences using
an arbitrary tie-breaking rule.\footnote{For example, break ties by setting $c_{i}\succ^{\theta|\mu}c_{j}$
whenever $u^{\theta}\left(c_{i};s\right)=u^{\theta}\left(c_{j};s\right)$
and $i<j$.} This is without loss, because any matching that is stable given the
refined preferences is also stable given the unrefined preferences. 

Throughout the paper, we impose the following standard assumption,
which corresponds to colleges having strict preferences over students. This assumption entails no loss of generality when colleges' priority rankings contain ties: if ties are broken arbitrarily, any matching that is stable under the resulting strict priorities is also stable under the original priorities.

\begin{assumption}
\label{ass:(Strict-preferences)}For any college $c\in\mathcal{C}$
and priority $x\in\mathbb{R}$ the set of students that have priority
equal to $x$ is of $\eta$-measure zero, that is, $\eta\left(\left\{ \theta\in\Theta\mid r_{c}^{\theta}=x\right\} \right)=0$.
\end{assumption}
Assumption \ref{ass:(Strict-preferences)} implies that students are
non-atomic, and, consequently, that the matching statistics are invariant to the
assignment of a single student. In other words, if two matchings $\mu,\mu'$
differ only in the assignment of student $\theta_{1}$ (that is, $\mu\left(\theta\right)=\mu'\left(\theta\right)$
for any $\theta\neq\theta_{1}$), then both matchings have identical
matching statistics (that is, $\sigma\left(\mu\right)=\sigma\left(\mu'\right)$).
Thus, given Assumption \ref{ass:(Strict-preferences)}, we can interpret $\succ^{\theta|\mu}$ as student $\theta$'s preferences over his potential college assignment given that the remaining students are assigned according to $\mu$. Given this, we define stable matchings as follows.
\begin{defn}
A student-college pair $\left(\theta,c\right)$ blocks matching $\mu$
if $c\succ^{\theta|\mu}\mu\left(\theta\right)$ and either $\eta\left(\mu\left(c\right)\right)<q_{c}$
or there is a student $\theta'\in\mu\left(c\right)$ such that $r_{c}^{\theta}>r_{c}^{\theta'}$.
A matching $\mu$ is stable if $\eta\left(\mu(c)\right)\leq q_{c}$
for all $c\in\mathcal{C}$ and there is no student-college pair $\left(\theta,c\right)$
that blocks $\mu$.
\end{defn}
In other words, a pair $\left(\theta,c\right)$ blocks the matching
$\mu$ if both would prefer to be matched to each other, holding the
rest of the match constant. Implicitly, this definition assumes that the endogenous student preferences $\succ^{\theta|\mu}$ of student $\theta$ are independent of the assignment of student $\theta$.
This guarantees that student $\theta$ has well-defined preferences
over potential assignments.\footnote{In particular, this avoids problems that can arise in discrete models.
For example, consider a discrete model with an extremely high GPA 
student who increases the average GPA of any college she is assigned to, and has a preference for the college with the highest GPA. Such
a student prefers whichever college she is assigned to.}

This notion of stability is a pairwise stability notion. That is,
this stability notion requires that \emph{pairs} of participants
cannot benefit by deviating from the match and forming alternative
matches. In standard matching models, pairwise stability is equivalent
to the core; that is, an arbitrary \emph{coalition} of participants
cannot benefit by deviating from the match and forming alternative
matches. In contrast, our pairwise stability notion is not equivalent
to the core, because a large coalition of agents might be able to
benefit by forming alternative matches and changing the matching statistics.\footnote{A coalition of students can benefit from forming an alternative match only if this alternative match has different matching statistics. Thus, a coalition of measure 0 of students cannot benefit by forming an alternative match.}
See Example \ref{exa:multiple-matchings} for an illustration.

\section{Characterization of Stable Matchings\label{sec:Cutoff-Char}}

We show in this section that any stable matching can be represented and solved for by cutoffs and matching statistics. This tractable representation provides a natural framework for empirical work, and it is used to prove the results in the following sections.

\emph{Cutoffs} $P\in\left[0,1\right]^{\mathcal{C}}$ specify a minimal
priority required for admission to each college $c\in\mathcal{C}$.
Given cutoffs $P$, student $\theta$ has sufficient priority to be
admitted to the set of colleges $B^{\theta}\left(P\right)=\left\{ c\in\mathcal{C}\,\mid\,r_{c}^{\theta}\geq P_{c}\right\} $.
We refer to $B^{\theta}\left(P\right)$ as $\theta$'s \emph{budget
set} given $P$. Given cutoffs $P$ and statistics $s$ we define
the \emph{demand} of student $\theta$ as\footnote{If $B^{\theta}\left(P\right)=\emptyset$ we define $D^{\theta}\left(P,s\right)=\phi$,
that is, $\theta$ demands to be unassigned.}
\[
D^{\theta}\left(P,s\right)=\max_{\succ^{\theta|s}}B^{\theta}\left(P\right)\,.
\]
Given $\left(P,s\right)\in\left[0,1\right]^{\mathcal{C}}\times\mathcal{X^{C}}$, the total mass of students that demands college $c\in\mathcal{C}$ is 
\[
D_{c}\left(P,s\mid\eta\right)=\eta\left(\left\{ \theta\mid D^{\theta}\left(P,s\right)=c\right\} \right)\,.
\]
 The \emph{demand }given $\left(P,s\right)$ is the vector 
\[
D\left(P,s\mid\eta\right)=\bigl(D_{c}\left(P,s\mid\eta\right)\bigr)_{c\in\mathcal{C}}\,,
\]
 specifying the mass of students that demand each college. We write
$D\left(P,s\right)$ when $\eta$ is clear from context.
\begin{defn}
\label{def:RE-MC}A pair $\left(P,s\right)\in\left[0,1\right]^{\mathcal{C}}\times\mathcal{X^{C}}$
of cutoffs and matching statistics is \emph{rational expectations
and market clearing }if
\begin{enumerate}
\item[(i)] \emph{Market Clearing:} 
\[
D\left(P,s\right)\leq q\,,
\]
and $P_{c}=0$ for any college $c$ for which $D_{c}\left(P,s\right)<q_{c}$.
\item[(ii)] \emph{Rational Expectations:} 
\[
s=\sigma\left(\mu\right)
\]
 for the matching $\mu$ defined by $\mu\left(\theta\right)=D^{\theta}\left(P,s\right)$. 
\end{enumerate}
\end{defn}
The matching $\mu\left(\theta\right)=D^{\theta}\left(P,s\right)$
can be implemented by publishing $\left(P,s\right)$. Each student
$\theta$ forms preferences over colleges $\succ^{\theta|s}$ given the matching statistics $s$. Using these preferences, each student $\theta$ chooses his most preferred college from the set of colleges $B^{\theta}\left(P\right)=\left\{ c\in\mathcal{C}\,\mid\,r_{c}^{\theta}\geq P_{c}\right\} $ for which $\theta$ has sufficient priority for admission given $P$. 
Condition (i) of Definition \ref{def:RE-MC} ensures that $\mu$ does not assign any college more students than the college's capacity, and that a college that did not fill its capacity is willing to admit any student. Condition (ii) ensures that students choose among colleges with the correct expectation of the matching statistics $s$. In other words, condition (ii) implies that if we assign each student $\theta'\neq\theta$ to $\mu\left(\theta'\right)$, student $\theta$'s ranking over his college potential assignments is indeed $\succ^{\theta|s}=\succ^{\theta|\mu}$.

The following lemma shows that stable matchings are equivalent to
rational expectations and market-clearing pair. 
\begin{lem}
\label{lem:(cutoff<->match)}If the pair $\left(P,s\right)$ is rational
expectations and market clearing, then the matching $\mu$ defined
by $\mu\left(\theta\right)=D^{\theta}\left(P,s\right)$ is a stable
matching.

Conversely, if $\mu$ is a stable matching then the pair $\left(P,s\right)$
defined by 
\begin{align*}
s=\sigma\left(\mu\right),\,\,\,\,\,\,\,\,\,\, &  & \begin{array}{cc}
P_{c} & =\begin{cases}
\inf\left\{ r_{c}^{\theta}\mid\theta\in\mu(c)\right\}  & \mbox{ if }\eta\left(\mu\left(c\right)\right)=q_{c}\\
0 & \mbox{ if }\eta\left(\mu\left(c\right)\right)<q_{c}
\end{cases}\end{array}
\end{align*}
is rational expectations and market clearing. 
\end{lem}
Lemma \ref{lem:(cutoff<->match)} provides a tractable characterization of stable matchings. It allows the calculation of a stable matching as a solution of a system of equations. The proof of Theorem \ref{thm:(stable exists)} in the following section builds upon this tractable representation.

\section{Existence and Structure of Stable Matchings\label{sec:Existence-and-Structure}}

In this section, we show that a stable matching always exists in economies
that satisfy one additional assumption. To motivate this assumption,
consider the following example. 
\begin{example}
\label{ex:Economy-without-stable-matching}There are two colleges
$\mathcal{C}=\left\{ c_{1},c_{2}\right\} $, each with capacity $q_{1}=q_{2}=1$.
All students have identical attributes and preferences and differ only by their priorities. Formally, the set of student types is $\Theta=\mathcal{X}\times\Gamma\times\left[0,1\right]^{2}=\left[0,1\right]\times\left\{ \bar{\gamma}\right\} \times\left[0,1\right]^{2}$.
The measure $\eta$ is uniform over $\left\{ \bar{\chi}\right\} \times\left\{ \bar{\gamma}\right\} \times\left[0,1\right]^{2}\subset\Theta$.
Students have a single attribute $\bar{\chi}=1$, and the matching statistics correspond to the mass of students assigned to each college, namely $s=\bigl(\eta\left(\mu\left(c_{1}\right)\right),\eta\left(\mu\left(c_{2}\right)\right)\bigr)$.
All students strictly prefer college $c_{1}$ over college $c_{2}$ if and only if the mass of students assigned to college $c_{1}$ is strictly below $1/2$, namely, the utility of any student $\theta$ is
\begin{align*}
u^{\theta}\left(c_{1};\mu\right) & =\begin{cases}
0 & \mbox{if }s_{c_{1}}\geq1/2\\
2 & \mbox{if }s_{c_{1}}<1/2\,
\end{cases}\\
u^{\theta}\left(c_{2};\mu\right) & =1\,.
\end{align*}
\end{example}

A stable matching does not exist in the economy of Example \ref{ex:Economy-without-stable-matching}, even though students are non-atomic. 
Because all agents have identical preferences and colleges have sufficient capacity, in any stable matching, all students must be matched to the same preferred college. Student preferences are discontinuous, and students always strictly prefer one of the colleges. But the collective mass of students behaves as a single discrete student with ``Groucho Marx preferences'' -- they would like to match to $c_{1}$ if and only if they are not matched to $c_{1}$.\footnote{\emph{\textquotedbl Please accept my resignation. I don't want to belong to any club that will accept people like me as a member\textquotedbl} - Groucho Marx in a telegram to the Friar's Club of Beverly Hills to which he belonged (\citealt{marx1959groucho}, p. 321).}

As discussed in the literature review, several papers rule out such economies by requiring that students' preferences are continuous. If the utility function in Example \ref{ex:Economy-without-stable-matching} was continuous there must have been some $s_{c_{1}}$ that makes all students indifferent between both colleges, and there would have existed a stable matching in which indifferent students are split between the colleges in precisely the correct ratio to make all students indifferent between the two colleges. This stable matching may be unsatisfactory because it crucially relies on all students being indifferent between the two colleges, and the implementation of this matching requires a precise split of these indifferent students between the two colleges.

We take a different approach. Rather than imposing restrictions on any individual student's preferences, we impose an aggregate condition on the distribution of preferences. 
In particular, our argument does not rely on the existence of students who are indifferent between colleges.\footnote{Assumption \ref{ass:(Continuity with Stats)} is related in spirit to several regularity conditions in the literature, but differs in what it restricts. 
The large-market regularity condition of \citet{liu2016ordinal} requires the assignment mechanism to produce similar allocations when the reports of a small fraction of participants change. 
By contrast, Assumption \ref{ass:(Continuity with Stats)} restricts agents' preferences rather than the sensitivity of the mechanism's outcome; in particular, it allows a stable matching to change substantially following changes in the reports of even a small number of agents. 
The assumptions of \citet{fisher2016matching} require that changing the assignment of a single agent not alter any other agent's ordinal ranking. By contrast, Assumption \ref{ass:(Continuity with Stats)} permits the ordinal preferences of a small mass of agents to change arbitrarily.}

\begin{assumption}[Diversity of preferences]
\label{ass:(Continuity with Stats)}For any matching statistics $s$,
and any $\varepsilon>0$ there exists $\delta>0$ such that $\eta\left(\left\{ \theta\,\mid\,\succ^{\theta|s}\neq\succ^{\theta|s'}\right\} \right)<\varepsilon$
for any $s'$ such that $\left\Vert s-s'\right\Vert _{\infty}<\delta$.%
\end{assumption}
In other words, diversity of preferences is satisfied if a sufficiently small change to the matching statistics causes only a small fraction of students to change their preference orderings over colleges. Diversity of preferences is trivially satisfied if student preferences do not depend on the matching statistics. In Section \ref{sec:App-econ-w-taste-shocks}, we show that the assumption is satisfied by standard empirical models that incorporate individual taste shocks.  

As an illustration, we adapt the economy from Example \ref{ex:Economy-without-stable-matching} to satisfy the diversity of preferences. 
\begin{example}
\label{ex:Economy-diverse-pref}We change the economy of Example \ref{ex:Economy-without-stable-matching}
by setting $\Gamma=\left[0,1\right]$, and taking the measure $\eta$
to be uniform\footnote{This measure corresponds to no correlation between a student's preferences
and priorities.} over $\left\{ \bar{\chi}\right\} \times\left[0,1\right]\times\left[0,1\right]^{2}\subset\Theta$.
The utility of $\theta=\left(\bar{\chi},\gamma^{\theta},r^{\theta}\right)$
is 
\begin{align*}
u^{\theta}\left(c_{1};\mu\right) & =\begin{cases}
0 & \mbox{if }s_{c_{1}}\geq\gamma^{\theta}\\
2 & \mbox{if }s_{c_{1}}<\gamma^{\theta}
\end{cases}\\
u^{\theta}\left(c_{2};\mu\right) & =1\,.
\end{align*}
 
\end{example}

The economy in Example \ref{ex:Economy-diverse-pref} satisfies diversity of preferences, as a small change to the matching statistics $s=\sigma\left(\mu\right)$ can cause a student $\theta$ to change his preference ordering only if $\gamma^{\theta}$ is sufficiently close to $s_{c_{1}}$. Although all students have discontinuous utility functions, because different students have different preferences, there exists an assignment that can satisfy all students. The economy has a unique stable matching $\mu$ given by
\[
\mu\left(\theta\right)=\begin{cases}
c_{1} & \mbox{if }\gamma^{\theta}\ge 1/2\\
c_{2} & \mbox{if }\gamma^{\theta}<1/2\,.
\end{cases}
\]
 The following theorem shows that the existence of a stable matching is not a coincidence and that a stable matching exists in any economy satisfying diversity of preferences.

\begin{thm}
\label{thm:(stable exists)}A stable matching exists in any economy
$E=[\Theta,C,\eta,q]$ that satisfies diversity of preferences. 
\end{thm}

The proof of Theorem \ref{thm:(stable exists)} relies on the cutoff characterization of stable matching. Lemma \ref{lem:(cutoff<->match)} allows us to transform the problem of finding a stable matching into the simpler finite-dimensional problem of finding $\left(P,s\right)$ that are market-clearing and rational expectations. After this transformation, we can use a fixed-point argument to show that there exists a market-clearing and rational expectations $\left(P,s\right)$. The diversity of preferences assumption ensures that the demand $D\left(P,s\right)$ is continuous in $P$ and $s$.

\subsection{The structure of the set of stable matchings}

In standard economies, the set of stable matchings forms a lattice, with a corresponding lattice of market-clearing cutoffs. In particular, there exists a stable matching that is preferred by all students, in that it gives each student higher utility than any other stable matching. Moreover, generically, there is a unique stable matching \citep{azevedo2016supply}. Previous work shows that these properties need not extend to economies with peer-dependent preferences \citep{sasaki1996two,pycia2019matching}. The following example demonstrates that, even in a continuum economy, a simple and natural peer externality can cause all of these properties to fail.

\begin{example}
\label{exa:multiple-matchings}The economy has two colleges $\mathcal{C}=\left\{ c_{1},c_{2}\right\} $,
each with capacity $q_{c_{1}}=q_{c_{2}}=1/2$. Each student type $\theta=\left(\chi^{\theta},\gamma^{\theta},r^{\theta}\right)$
has a single attribute $\chi^{\theta}\in\left[0,1\right]=\mathcal{X}$
which represents the student's GPA (normalized to be between 0 and
1). The matching statistics $s=\left(s_{c_{1}},s_{c_{2}}\right)$
correspond to the sum of student GPAs at each college. Student utility
is

\[
u^{\theta}\left(c;s\right)=g_{c}^{\theta}+\beta^{\theta}\cdot s_{c}\,,
\]
which is parameterized by $\gamma^{\theta}=\left(g_{c_{1}}^{\theta},g_{c_{2}}^{\theta},\beta^{\theta}\right)$.
That is, the student gains utility $g_{c}^{\theta}$ from attending
college $c$ irrespective of peers, and gains utility $\beta^{\theta}\cdot s_{c}$
from the quality of his peers. The measure $\eta$ is uniform over
\[ 
\left\{ \left(\chi^{\theta},\gamma^{\theta},r^{\theta}\right)\,\middle|\,
\chi^{\theta}=\beta^{\theta}/16=r_{c_{1}}^{\theta}=r_{c_{2}}^{\theta},\;
 g_{c_{1}}^{\theta}=1,\;g_{c_{2}}^{\theta}=2\right\}\subset\Theta.
\]
Thus, student GPA is uniformly distributed, $\chi^{\theta}\sim U\left[0,1\right]$,
and the student type $\theta=\left(\chi^{\theta},\gamma^{\theta},r^{\theta}\right)$
is fully determined by the student's GPA. In particular, the student's
rank at both colleges equals the student's GPA,
$r_{c_{1}}^{\theta}=r_{c_{2}}^{\theta}=\chi^{\theta}$, and the preference
parameters are $\beta^{\theta}=16\chi^{\theta}$,
$g_{c_{1}}^{\theta}=1$, and $g_{c_{2}}^{\theta}=2$.
\end{example}
The economy of Example \ref{exa:multiple-matchings} features students who prefer colleges that are attended by peers with higher GPAs.\footnote{Note that student utilities specify the student's preferences over colleges, which may differ from the student's educational gains at each college (e.g., \citet{abdulkadirouglu2020parents}).} Students with a GPA higher than $1/2$ are assigned to their most preferred college under any stable matching, because such students have sufficiently high priority at both colleges. Although $g_{c_{2}}^{\theta}>g_{c_{1}}^{\theta}$, a high GPA student may prefer college $c_{1}$ if sufficiently many high GPA students are assigned to college $c_{1}$. As a result, these high-GPA students face a coordination problem.

The economy of Example \ref{exa:multiple-matchings} has three stable
matchings 
\[
\begin{array}{ccc}
\mu^{\dagger}\left(\theta\right)=\begin{cases}
c_{1} & \chi^{\theta}\geq1/2\\
c_{2} & \chi^{\theta}<1/2
\end{cases}, & \mu^{\vartriangle}\left(\theta\right)=\begin{cases}
c_{2} & a^{\vartriangle}-\frac{1}{2}\leq\chi^{\theta}<a^{\vartriangle}\\
c_{1} & \text{otherwise}
\end{cases}, & \mu^{\circ}\left(\theta\right)=\begin{cases}
c_{2} & \chi^{\theta}\geq1/2\\
c_{1} & \chi^{\theta}<1/2
\end{cases}\end{array}
\]
where $a^{\vartriangle}=\left(3+\sqrt{5}\right)/8\approx0.65$. In
matching $\mu^{\dagger}$ high GPA students coordinate on college
$c_{1}$. In matching $\mu^{\circ}$ high GPA students coordinate
on college $c_{2}$. In matching $\mu^{\vartriangle}$ the GPAs at
the two colleges are less differentiated. Students with $\chi^{\theta}>a^{\vartriangle}$
prefer college $c_{1}$, students with $\chi^{\theta}<a^{\vartriangle}$
prefer college $c_{2}$, and the type $\chi^{\theta}=a^{\vartriangle}$
is indifferent. The corresponding
cutoffs and statistics are $P^{\dagger}=\left(\frac{1}{2},0\right),s^{\dagger}=\left(\frac{3}{8},\frac{1}{8}\right)$,
$P^{\vartriangle}=\left(0,a^{\vartriangle}-\frac{1}{2}\right),s^{\vartriangle}=\left(\frac{5}{8}-\frac{1}{2}a^{\vartriangle},\frac{1}{2}a^{\vartriangle}-\frac{1}{8}\right)\approx\left(0.3,0.2\right)$,
and $P^{\circ}=\left(0,\frac{1}{2}\right),s^{\circ}=\left(\frac{1}{8},\frac{3}{8}\right)$.
Appendix \ref{sec:Omitted-Calculations} provides the calculations.

In contrast to standard matching markets, there is no stable matching
that is preferred by all students. Students with high GPA,
$\chi^{\theta}\in\left[1/2,1\right]$, prefer $\mu^{\circ}$ over
$\mu^{\dagger}$ over $\mu^{\vartriangle}$, namely,
\[
u^{\theta}\left(\mu^{\circ}\left(\theta\right);s^{\circ}\right)
>u^{\theta}\left(\mu^{\dagger}\left(\theta\right);s^{\dagger}\right)
>u^{\theta}\left(\mu^{\vartriangle}\left(\theta\right);s^{\vartriangle}\right).
\]
Students with intermediate GPA,
$\chi^{\theta}\in\left[a^{\vartriangle}-\frac{1}{2},\frac{1}{2}\right)$,
have the reverse ranking, namely,
\[
u^{\theta}\left(\mu^{\circ}\left(\theta\right);s^{\circ}\right)
<u^{\theta}\left(\mu^{\dagger}\left(\theta\right);s^{\dagger}\right)
<u^{\theta}\left(\mu^{\vartriangle}\left(\theta\right);s^{\vartriangle}\right).
\]
Thus, high-GPA students most
prefer $\mu^{\circ}$, in which they attend college $c_{2}$ and benefit
from both its higher intrinsic value and stronger peers, whereas
intermediate-GPA students most prefer $\mu^{\vartriangle}$, in which
they attend college $c_{2}$ with a mix of higher- and lower-GPA students. 

In addition, the cutoffs do not form a lattice. For example, there is no market-clearing rational expectation pair $\left(P,s\right)$ with $P=\left(\frac{1}{2},\frac{1}{2}\right)$ or $P=\left(0,0\right)$.
 
Example \ref{exa:multiple-matchings} also illustrates that our pairwise-stability notion is not equivalent to the core; that is, a coalition may benefit by forming alternative matches. To see this, observe that given $\mu^{\dagger}$ students with high GPA ($\chi^{\theta}\in\left[1/2,1\right]$) and college $c_{2}$ can benefit by collectively deviating and matching to each other.

Example \ref{exa:multiple-matchings} also shows that a highly demanded college may not necessarily have the best intrinsic qualities. In matching $\mu^{\dagger}$, college $c_{1}$ has a higher cutoff and is preferred by all students with $\chi^{\theta}>1/4$, despite all students agreeing that $g_{c_{2}}^{\theta}>g_{c_{1}}^{\theta}$. If matching $\mu^{\dagger}$ is realized, analysis of elicited preferences that does not account for the possibility of peer-dependent preferences can erroneously indicate that students have strong preference for $c_{1}$ because of its intrinsic quality.\footnote{Holding the matching statistics $s$ fixed, the true model with $u^{\theta}\left(c;s\right)=g_{c}^{\theta}+\beta^{\theta}\cdot s_{c}$ is indistinguishable from a model in which students have preferences $u^{\theta}\left(c;s\right)=\tilde{g}_{c}^{\theta}$ that only depend on a college intrinsic qualities $\tilde{g}_{c}^{\theta}$ with $\tilde{g}_{c}^{\theta}=g_{c}^{\theta}+\beta^{\theta}\cdot s_{c}$.} 

The distinction between peer-dependent preferences and preferences for intrinsic school quality can have important policy implications. For example, consider a policy that transfers all students from college $c_{1}$ to a new college. This policy is harmful if students have a strong preference for college $c_{1}$ because of its intrinsic quality. But the policy can lead to a Pareto improvement if students have a preference for college $c_{1}$ only due to peer effects, as in Example \ref{exa:multiple-matchings}. 

In standard matching markets, such multiplicity of stable matching is a knife-edge case \citep{azevedo2016supply}. In contrast, the multiplicity of stable matchings holds in many variants of the economy of Example \ref{exa:multiple-matchings} and does not require precisely tuned parameters.

As in standard matching markets, it is possible for a stable matching to be Pareto dominated for students by another stable matching, in the sense that all students strictly prefer one stable matching over another. Appendix \ref{sec:Additional-Examples} provides such an example. 

Finally, the rural hospital theorem for standard matching markets states that a student who is unassigned under some stable matching must be unassigned under all stable matchings. Appendix \ref{sec:Additional-Examples} provides an example showing that the rural hospital theorem fails to hold in matching markets with peer-dependent preferences. An additional example in Appendix \ref{sec:Additional-Examples} shows it is possible to have different stable matchings that have the same cutoffs. 

\section{Approximate Stability in Finite Sampled Economies\label{sec:Approximately-Stability}}

A continuum economy $E=[\Theta,\mathcal{C},\eta,q]$ can be interpreted
as describing the population of students from which finite economies
are sampled. We show that, under mild regularity assumptions, sampled economies admit an approximately stable matching. 

Formally, a discrete economy with peer-dependent preferences can be
captured by $F=[\Theta,\mathcal{C},\hat{\eta},\hat{q}]$ where $\hat{q}\in\mathbb{N}^{\mathcal{C}}$
and the measure $\hat{\eta}$ is a discrete measure with distinct
atoms of equal measure, each atom corresponding to a student. A sampled
economy $F^{n}=[\Theta,\mathcal{C},\hat{\eta}^{n},\hat{q}^{n}]$ from
continuum economy $E=[\Theta,\mathcal{C},\eta,q]$ is generated by
discretizing capacities $\hat{q}^{n}=\left\lfloor n\cdot q\right\rfloor /n$
and independently sampling $n$ students $\theta_{1},\dots,\theta_{n}$
from $\eta$. That is, $\hat{\eta}$ is a discrete measure with $\hat{\eta}\left(A\right)=\frac{1}{n}\cdot\left|A\cap\left\{ \theta_{1},\dots,\theta_{n}\right\} \right|$
for any $A\subset\Theta$. The matching statistics $\sigma\left(\hat{\mu}^{n}\right)$
given a finite matching $\hat{\mu}^{n}:\Theta\to\mathcal{C}\cup\left\{ \phi\right\} $
are given by 
\[
\sigma_{c}\left(\hat{\mu}^{n}\right)=\frac{1}{n}\sum_{\theta_{i}\in\mu\left(c\right)}\chi^{_{\theta_{i}}},
\]
which is the finite analogue of equation (\ref{eq:match-stats}).

Stable matchings may fail to exist in finite economies. We define a matching to be approximately stable if any student who can block the match can only gain a small amount of utility from blocking. If the student's utility gain from blocking is small, small frictions may suffice to prevent the formation of blocking pairs and prevent unraveling.
\begin{defn}
\label{def:approx-stable}A matching $\mu$ with statistics $s=\sigma\left(\mu\right)$
is $\varepsilon$-stable if $\eta\left(\mu\left(c\right)\right)\leq q_{c}$
for all $c$, and for any student-college pair $\left(\theta,c\right)$
that blocks matching $\mu$ the student's gain from blocking is $u^{\theta}\left(c;s\right)-u^{\theta}\left(\mu\left(\theta\right);s\right)<\varepsilon$. 
\end{defn}
Our analysis so far imposed no restriction on students' cardinal utility $u^{\theta}\left(c;s\right)$. Moreover, stability depends only on ordinal preferences $\succ^{\theta|s}$, which are invariant to monotone transformations of the utility function. Therefore, obtaining meaningful results about the magnitude of a student's utility gain from blocking requires an additional assumption.
\begin{defn}
\label{assu:pref-cont} Utilities are \emph{$\alpha$ locally peer-smooth
at matching statistics $s$} if for any student $\theta\in\Theta$,
college $c$ and matching statistics $s'$ it holds that $\left|u^{\theta}\left(c;s\right)-u^{\theta}\left(c;s'\right)\right|\leq\alpha\left\Vert s-s'\right\Vert _{\infty}$.
Utilities are $\alpha$ peer-smooth if they are $\alpha$ locally
peer-smooth at any matching statistics $s$.
\end{defn}
Section \ref{sec:App-econ-w-taste-shocks} shows that utility functions used in standard empirical models are $\alpha$ peer-smooth.

\begin{thm}
\label{thm:aprox-stable-for-sampled}Let $E=[\Theta,\mathcal{C},\eta,q]$
be a continuum economy, and let $\left(P^{*},s^{*}\right)$ be rational
expectations and market clearing such that $\partial_{P}D\left(P^{*},s^{*}\right)$
exists and is invertible, and utilities are $\alpha$ peer-smooth
at $s^{*}$ for some $\alpha>0$. \\ 
Let $F^{n}=[\Theta,\mathcal{C},\eta^{n},q^{n}]$ be a sequence of finite economies sampled from $E$. Then for any $\varepsilon>0$ there exists $N$ such that for any $n>N$, with probability of at least $1-\varepsilon$, the sampled economy $F^{n}$ has an $\varepsilon$-stable matching. Moreover, there is a $\varepsilon$-stable matching of $F^{n}$ that corresponds to $\left(P,s\right)$ such that 
$\|P-P^{*}\|_{\infty} < \varepsilon$ and 
$\|s-s^{*}\|_{\infty} < \frac{\varepsilon}{2\alpha}$. 
\end{thm}

The proof of Theorem \ref{thm:aprox-stable-for-sampled} leverages the convergence of standard economies\footnote{The proof extends the techniques used by \citet{azevedo2016supply}, as it does not directly follow from previous results.} to find cutoffs that are close to $P^{*}$ and matching statistics that are close to $s^{*}$ that induce an $\varepsilon$-stable matching.
It may be necessary to adjust both the cutoffs and the matching statistics, because in an $\varepsilon$-stable matching, no college can be assigned more students than its capacity.

Theorem \ref{thm:aprox-stable-for-sampled} implies that an $\varepsilon$-stable
matching exists with high probability in a discrete economy sampled
from the economy in Example \ref{exa:multiple-matchings}. Moreover, for any of the three stable matchings in that economy, there exists an $\varepsilon$-stable matching with approximately the same cutoffs
and statistics.
\begin{rem*}
The requirement that utilities are $\alpha$ locally peer-smooth implies
local continuity of preferences. A weaker condition that does not require continuity of preferences can be used to show the existence of
an $\varepsilon$-stable matching for a fixed $\varepsilon$. With
minor changes to the proof of Theorem \ref{thm:aprox-stable-for-sampled}
it can be shown that if there exists $\varepsilon_{u},\delta>0$ such
that the (potentially discontinuous) utilities satisfy $\left|u^{\theta}\left(c;s\right)-u^{\theta}\left(c;s'\right)\right|\leq\varepsilon_{u}$
when $\left\Vert s-s'\right\Vert _{\infty}<\delta$, then for any
$\varepsilon_{p}>0$ there exists $N$ such for any $n>N$ the sampled
economy $F^{n}$ has an $\varepsilon_{u}$-stable matching with probability
of at least $1-\varepsilon_{p}$.
\end{rem*}

\section{ Existences of Stable Matching in Common Empirical Specifications with Peer Effects \label{sec:App-econ-w-taste-shocks}}

This section shows that a class of random-utility specifications commonly used in empirical analyses of school choice fits naturally within our framework. These specifications allow students to value both intrinsic school characteristics and their peers' composition. We show that continuously distributed idiosyncratic taste shocks imply that diversity of preferences is satisfied, and therefore guarantee the existence of a stable matching in the continuum economy. Under the additional local regularity condition imposed in Theorem \ref{thm:aprox-stable-for-sampled}, sufficiently large finite economies sampled from the continuum economy also admit approximately stable matchings with high probability.

A student of type
\[
\theta=\left(\chi^{\theta},\gamma^{\theta},r^{\theta}\right)
\]
has utility
\[
u^{\theta}\left(c;s\right)
=
g_{c}^{\theta}
+
\beta^{\theta}\cdot s_{c}
+
\xi_{c}^{\theta},
\]
where
\[
\gamma^\theta
=
\left(g^{\theta},\beta^{\theta},\xi^{\theta}\right)
=
\left(
g_{c_{1}}^{\theta},\ldots,g_{c_{J}}^{\theta},
\beta_{1}^{\theta},\ldots,\beta_{L}^{\theta},
\xi_{c_{1}}^{\theta},\ldots,\xi_{c_{J}}^{\theta}
\right).
\]
Here, \(g_c^\theta\) captures the student's utility from intrinsic school characteristics, \(\beta^\theta\) governs the value placed on different peer characteristics, and \(\xi_c^\theta\) is an idiosyncratic school-specific taste shock. Specifications of this form are widely used in empirical work on school choice and related settings (e.g.,
\citealp{epple2018superintendent,allende2019competition}).

We allow arbitrary correlation among student attributes $\chi^{\theta}$, priorities $r^{\theta}$, intrinsic utilities $g^{\theta}$, and peer-preference coefficients $\beta^{\theta}$. The only distributional restriction is that the taste shocks are independent of these variables and independently distributed across colleges according to a common distribution $F$ with a bounded density $f$. Formally, for measurable sets $X\subset\mathcal{X}, G\subset\mathbb{R}^{\mathcal{C}}, B\subset\mathbb{R}^{L}, R\subset\left[0,1\right]^{\mathcal{C}}$ and any rectangle $
V=\prod_{c\in\mathcal{C}}[a_{c},b_{c}]
\subseteq\mathbb{R}^{\mathcal{C}},
$
the measure $\eta$ satisfies
\begin{align*}
&\eta\left(
\left\{
\theta\ \middle|\ 
\chi^{\theta}\in X,\,
g^{\theta}\in G,\,
\beta^{\theta}\in B,\,
r^{\theta}\in R,\,
\xi^{\theta}\in V
\right\}
\right)
\\
&\qquad =
\eta\left(
\left\{
\theta\ \middle|\ 
\chi^{\theta}\in X,\,
g^{\theta}\in G,\,
\beta^{\theta}\in B,\,
r^{\theta}\in R
\right\}
\right)
\prod_{c\in\mathcal{C}}
\left[F(b_{c})-F(a_{c})\right].
\end{align*}
We refer to an economy satisfying these conditions as a \emph{linear-preferences economy}. The next lemma formally shows that these economies meet the assumptions required for the existence of a stable matching.

\begin{lem}
\label{lem:linear-economy-is-smooth} Let $E=[\Theta,\mathcal{C},\eta,q]$ be a linear-preferences economy with taste shocks distributed according to $F$. 
If $F$ has a bounded density $f$, then the economy satisfies diversity of preferences. In addition, if
\[
\beta_{\max}=\sup_{\theta\in\Theta}
\left\Vert\beta^{\theta}\right\Vert_{\infty}
<\infty,
\]
then cardinal utilities are $L\beta_{\max}$ peer-smooth.
\end{lem}

The intuition is that a small change in the matching statistics can change a student's ranking of two colleges only if the student was initially close to indifferent between them. Individual students may nevertheless have highly heterogeneous preferences over peer composition, and the model permits arbitrary correlation between these preferences, student attributes, and college priorities. Moreover,
\[
\left|
\beta^{\theta}\cdot(s_{c}-s'_{c})
\right|
\leq
L\beta_{\max}
\left\Vert s-s'\right\Vert_{\infty}.
\]
Thus, uniformly bounded coefficients on peer characteristics imply the peer-smoothness condition used to obtain approximate stability in sampled economies.

Combining Lemma \ref{lem:linear-economy-is-smooth} with the main existence results yields the following immediate implication for empirical specifications.

\begin{thm}
\label{thm:linear-economy}
Consider a linear-preferences economy with independent taste shocks whose common distribution admits a bounded density. Then the continuum economy has a stable matching.

In addition, suppose that $\beta_{\max}=\sup_{\theta\in\Theta}\left\Vert\beta^{\theta}\right\Vert_{\infty}<\infty$ and that there exists a rational-expectations and market-clearing pair $\left(P^{*},s^{*}\right)$ such that $\partial_{P}D\left(P^{*},s^{*}\right)$ exists and is invertible. Then, for every $\varepsilon>0$, sufficiently large finite economies sampled from the continuum economy have an $\varepsilon$-stable matching with probability at least $1-\varepsilon$. 
\end{thm}

The first part follows from Lemma \ref{lem:linear-economy-is-smooth} and Theorem \ref{thm:(stable exists)}. The second follows from Lemma \ref{lem:linear-economy-is-smooth} and Theorem \ref{thm:aprox-stable-for-sampled}. Thus, the existence results in this paper apply directly to commonly used empirical models that combine intrinsic school characteristics, linear preferences over peer composition, and continuously distributed taste shocks.

\section{Conclusion \label{sec:Conclusion}}

The model developed in this paper allows a general form of peer-dependent preferences, and does not require continuous preferences or rely on leaving students indifferent between colleges. The diversity of preferences condition implies that while individual students may be discontinuous, the demand given cutoffs and statistics is continuous. Using this observation and the tractable cutoff representation, we show the existence of a stable matching and the existence of approximately stable matchings for finite sampled economies.

Our results suggest that, under conditions relevant for empirical applications, finite-market nonexistence is largely due to rounding errors, and that approximate stability is a natural empirical solution concept. By contrast, other failures of classic matching results previously identified in the literature are not due to rounding errors. For example, multiplicity of stable matchings can occur even in the continuum when students have strong incentives to coordinate. 

Standard matching models can, of course, accommodate preferences that reflect peer quality as long as those preferences are treated as fixed. For example, an elite school may be attractive largely because it consistently enrolls strong students; if its peer composition is effectively fixed, it does not matter whether the model endogenizes the school's quality. Endogenizing peer-dependent preferences becomes important when evaluating counterfactuals that substantially change school composition and, therefore, may also change the school's attractiveness.

Even when such counterfactuals are not the primary object of interest, it may still be important to distinguish intrinsic school quality from peer-driven attractiveness. This may be challenging to study empirically; the examples in this paper show that observed demand may reveal little about the school's intrinsic quality. 

Even when such counterfactuals are not the primary object of interest, it may still be important to distinguish intrinsic school quality from peer-driven attractiveness. This distinction may be difficult to identify empirically; the examples in this paper show that observed demand may reveal little about a school's intrinsic quality.

Our framework rules out preferences that depend on the assignment of particular individuals. Allowing such dependence would encompass settings such as matching with couples and introduce a distinct source of complications. Existing large-market results for couples suggest that these complications may also become manageable in sufficiently large markets \citep{kojima2013matching,ashlagi2014stability,nguyen2018near}.

\printbibliography

@article{abdulkadirouglu2003school,
  title={School choice: A mechanism design approach},
  author={Abdulkadiro{\u{g}}lu, Atila and S{\"o}nmez, Tayfun},
  journal={American Economic Review},
  volume={93},
  number={3},
  pages={729--747},
  year={2003}
}

@article{abdulkadirouglu2005new,
  title={The New York City high school match},
  author={Abdulkadiro{\u{g}}lu, Atila and Pathak, Parag A and Roth, Alvin E},
  journal={American Economic Review},
  volume={95},
  number={2},
  pages={364--367},
  year={2005}
}

@article{jagadeesan2017complementary,
  title={Complementary Inputs and the Existence of Stable Outcomes in Large Trading Networks.},
  author={Jagadeesan, Ravi},
  journal={Proceedings of the 2017 ACM Conference on Economics and Computation},
  pages={265},
  year={2017}
}

@article{rostek2020matching,
  title={Matching with complementary contracts},
  author={Rostek, Marzena and Yoder, Nathan},
  journal={Econometrica},
  volume={88},
  number={5},
  pages={1793--1827},
  year={2020},
  publisher={Wiley Online Library}
}

@article{greinecker2018pairwise,
  title={Pairwise stable matching in large economies},
  author={Greinecker, Michael and Kah, Christopher},
  journal={Econometrica},
  volume={89},
  number={6},
  pages={2929--2974},
  year={2021},
  publisher={Wiley Online Library}
}

@article{azevedo2018existence,
  title={Existence of equilibrium in large matching markets with complementarities},
  author={Azevedo, Eduardo M and Hatfield, John William},
  journal={Available at SSRN 3268884},
  year={2018}
}

@article{che2019stable,
  title={Stable matching in large economies},
  author={Che, Yeon-Koo and Kim, Jinwoo and Kojima, Fuhito},
  journal={Econometrica},
  volume={87},
  number={1},
  pages={65--110},
  year={2019},
  publisher={Wiley Online Library}
}

@article{schiltz2019does,
  title={Does it matter when your smartest peers leave your class? Evidence from Hungary},
  author={Schiltz, Fritz and Mazrekaj, Deni and Horn, Daniel and De Witte, Kristof},
  journal={Labour Economics},
  volume={59},
  pages={79--91},
  year={2019},
  publisher={Elsevier}
}

@article{pycia2019matching,
  title={Matching with externalities},
  author={Pycia, Marek and Yenmez, M Bumin},
  journal={The Review of Economic Studies},
  volume={90},
  number={2},
  pages={948--974},
  year={2023},
  publisher={Oxford University Press US}
}

@article{bando2012many,
  title={Many-to-one matching markets with externalities among firms},
  author={Bando, Keisuke},
  journal={Journal of Mathematical Economics},
  volume={48},
  number={1},
  pages={14--20},
  year={2012},
  publisher={Elsevier}
}

@article{sasaki1996two,
  title={Two-sided matching problems with externalities},
  author={Sasaki, Hiroo and Toda, Manabu},
  journal={Journal of Economic Theory},
  volume={70},
  number={1},
  pages={93--108},
  year={1996},
  publisher={Elsevier}
}

@article{mumcu2010stable,
  title={Stable one-to-one matchings with externalities},
  author={Mumcu, Ay{\c{s}}e and Saglam, Ismail},
  journal={Mathematical Social Sciences},
  volume={60},
  number={2},
  pages={154--159},
  year={2010},
  publisher={Elsevier}
}

@article{calsamiglia2020school,
  title={School choice design, risk aversion and cardinal segregation},
  author={Calsamiglia, Caterina and Mart{\'\i}nez-Mora, Francisco and Miralles, Antonio},
  journal={The Economic Journal},
  volume={131},
  number={635},
  pages={1081--1104},
  year={2021},
  publisher={Oxford University Press}
}

@article{baccara2012field,
  title={A field study on matching with network externalities},
  author={Baccara, Mariagiovanna and {\.I}mrohoro{\u{g}}lu, Ay{\c{s}}e and Wilson, Alistair J and Yariv, Leeat},
  journal={American Economic Review},
  volume={102},
  number={5},
  pages={1773--1804},
  year={2012}
}

@article{fisher2016matching,
  title={Matching with aggregate externalities},
  author={Fisher, James CD and Hafalir, Isa E},
  journal={Mathematical Social Sciences},
  volume={81},
  pages={1--7},
  year={2016},
  publisher={Elsevier}
}

@article{lavy2011mechanisms,
  title={Mechanisms and impacts of gender peer effects at school},
  author={Lavy, Victor and Schlosser, Analia},
  journal={American Economic Journal: Applied Economics},
  volume={3},
  number={2},
  pages={1--33},
  year={2011}
}

@article{lavy2012inside,
  title={Inside the black box of ability peer effects: Evidence from variation in the proportion of low achievers in the classroom},
  author={Lavy, Victor and Paserman, M Daniele and Schlosser, Analia},
  journal={The Economic Journal},
  volume={122},
  number={559},
  pages={208--237},
  year={2012},
  publisher={Oxford University Press Oxford, UK}
}

@article{booij2017ability,
  title={Ability peer effects in university: Evidence from a randomized experiment},
  author={Booij, Adam S and Leuven, Edwin and Oosterbeek, Hessel},
  journal={The Review of Economic Studies},
  volume={84},
  number={2},
  pages={547--578},
  year={2017},
  publisher={Oxford University Press}
}

@article{duflo2011peer,
  title={Peer effects, teacher incentives, and the impact of tracking: Evidence from a randomized evaluation in Kenya},
  author={Duflo, Esther and Dupas, Pascaline and Kremer, Michael},
  journal={American economic review},
  volume={101},
  number={5},
  pages={1739--74},
  year={2011}
}

@article{campos2022impact,
  title={The Impact of Public School Choice: Evidence from Los Angeles''s Zones of Choice},
  author={Campos, Christopher and Kearns, Caitlin},
  journal={The Quarterly Journal of Economics},
  volume={139},
  number={2},
  pages={1051--1093},
  year={2024},
  publisher={Oxford University Press}
}

@article{agarwal2020revealed,
  title={Revealed preference analysis of school choice models},
  author={Agarwal, Nikhil and Somaini, Paulo},
  journal={Annual Review of Economics},
  volume={12},
  pages={471--501},
  year={2020},
  publisher={Annual Reviews}
}

@inproceedings{Cox2022peer,
  title={Peer preferences in centralized school choice markets: theory and evidence},
  author={Cox, Natalie and Fonseca, Ricardo and Pakzad-Hurson, Bobak and Pecenco, Matthew},
  year={2023},
  institution={working paper}
}

@article{barseghyan2019peer,
  title={Peer Preferences, School Competition, and the Effects of Public School Choice},
  author={Barseghyan, Levon and Clark, Damon and Coate, Stephen},
  journal={American Economic Journal: Economic Policy},
  volume={11},
  number={4},
  pages={124--58},
  year={2019}
}

@article{abdulkadirouglu2015expanding,
  title={Expanding {`choice'} in school choice},
  author={Abdulkadiro{\u{g}}lu, Atila and Che, Yeon-Koo and Yasuda, Yosuke},
  journal={American Economic Journal: Microeconomics},
  volume={7},
  number={1},
  pages={1--42},
  year={2015}
}

@article{azevedo2013walrasian,
  title={Walrasian equilibrium in large, quasilinear markets},
  author={Azevedo, Eduardo M and Weyl, E Glen and White, Alexander},
  journal={Theoretical Economics},
  volume={8},
  number={2},
  pages={281--290},
  year={2013},
  publisher={Wiley Online Library}
}

@article{liu2014stable,
  title={Stable matching with incomplete information},
  author={Liu, Qingmin and Mailath, George J and Postlewaite, Andrew and Samuelson, Larry},
  journal={Econometrica},
  volume={82},
  number={2},
  pages={541--587},
  year={2014},
  publisher={Wiley Online Library}
}

@article{rothstein2006good,
  title={Good principals or good peers? Parental valuation of school characteristics, Tiebout equilibrium, and the incentive effects of competition among jurisdictions},
  author={Rothstein, Jesse M},
  journal={American Economic Review},
  volume={96},
  number={4},
  pages={1333--1350},
  year={2006}
}

@article{bianchi2020indirect,
  title={The Indirect Effects of Educational Expansions: Evidence from a Large Enrollment Increase in University Majors},
  author={Bianchi, Nicola},
  journal={Journal of Labor Economics},
  volume={38},
  number={3},
  pages={767--804},
  year={2020},
  publisher={The University of Chicago Press Chicago, IL}
}

@article{ashlagi2014stability,
  title={Stability in large matching markets with complementarities},
  author={Ashlagi, Itai and Braverman, Mark and Hassidim, Avinatan},
  journal={Operations Research},
  volume={62},
  number={4},
  pages={713--732},
  year={2014},
  publisher={INFORMS}
}

@article{kojima2013matching,
  title={Matching with couples: Stability and incentives in large markets},
  author={Kojima, Fuhito and Pathak, Parag A and Roth, Alvin E},
  journal={The Quarterly Journal of Economics},
  volume={128},
  number={4},
  pages={1585--1632},
  year={2013},
  publisher={MIT Press}
}

@article{chakraborty2010two,
  title={Two-sided matching with interdependent values},
  author={Chakraborty, Archishman and Citanna, Alessandro and Ostrovsky, Michael},
  journal={Journal of Economic Theory},
  volume={145},
  number={1},
  pages={85--105},
  year={2010},
  publisher={Elsevier}
}

@article{klaus2005stable,
  title={Stable matchings and preferences of couples},
  author={Klaus, Bettina and Klijn, Flip},
  journal={Journal of Economic Theory},
  volume={121},
  number={1},
  pages={75--106},
  year={2005},
  publisher={Elsevier}
}

@article{ashlagi2014improving,
  title={Improving community cohesion in school choice via correlated-lottery implementation},
  author={Ashlagi, Itai and Shi, Peng},
  journal={Operations Research},
  volume={62},
  number={6},
  pages={1247--1264},
  year={2014},
  publisher={INFORMS}
}

@article{dur2019school,
  title={School choice with neighbors},
  author={Dur, Umut Mert and Wiseman, Thomas},
  journal={Journal of Mathematical Economics},
  volume={83},
  pages={101--109},
  year={2019},
  publisher={Elsevier}
}

@article{pycia2012stability,
  title={Stability and preference alignment in matching and coalition formation},
  author={Pycia, Marek},
  journal={Econometrica},
  volume={80},
  number={1},
  pages={323--362},
  year={2012},
  publisher={Wiley Online Library}
}

@article{echenique2007solution,
  title={A solution to matching with preferences over colleagues},
  author={Echenique, Federico and Yenmez, M Bumin},
  journal={Games and Economic Behavior},
  volume={59},
  number={1},
  pages={46--71},
  year={2007},
  publisher={Elsevier}
}

@article{abdulkadirouglu2020parents,
  title={Do parents value school effectiveness?},
  author={Abdulkadiro{\u{g}}lu, Atila and Pathak, Parag A and Schellenberg, Jonathan and Walters, Christopher R},
  journal={American Economic Review},
  volume={110},
  number={5},
  pages={1502--39},
  year={2020}
}

@article{dworczak2016deferred,
  title={Deferred acceptance with compensation chains},
  author={Dworczak, Piotr},
  journal={Proceedings of the 2016 ACM Conference on Economics and Computation},
  pages={65--66},
  year={2016}
}

@article{blum1997vacancy,
  title={Vacancy chains and equilibration in senior-level labor markets},
  author={Blum, Yosef and Roth, Alvin E and Rothblum, Uriel G},
  journal={Journal of Economic Theory},
  volume={76},
  number={2},
  pages={362--411},
  year={1997},
  publisher={Elsevier}
}

@article{vapnik1972uniform,
  title={On the uniform convergence of relative frequencies of events to their probabilities},
  author={Vapnik, Vladimir N and Chervonenkis, Alexey Y},
  editor={Vovk V., Papadopoulos H., Gammerman A.},
  journal={Measures of Complexity},
  publisher={Springer, Cham},
  pages={11--30},
  year={1972}
}

@article{epple2018superintendent,
  title={The superintendent's dilemma: Managing school district capacity as parents vote with their feet},
  author={Epple, Dennis and Jha, Akshaya and Sieg, Holger},
  journal={Quantitative Economics},
  volume={9},
  number={1},
  pages={483--520},
  year={2018},
  publisher={Wiley Online Library}
}

@article{epple1998competition,
  title={Competition between private and public schools, vouchers, and peer-group effects},
  author={Epple, Dennis and Romano, Richard E},
  journal={American Economic Review},
  pages={33--62},
  year={1998},
  publisher={JSTOR}
}

@article{allende2019competition,
  title={Competition Under Social Interactions and the Design of Education Policies},
  author={Allende, Claudia},
  journal={Working Paper},
  year={2021}
}

@book{marx1959groucho,
  title={Groucho and me: the autobiography of Groucho Marx},
  author={Marx, Groucho},
  year={1959},
  publisher={Victor Gollancz, London}
}

@article{azevedo2016supply,
  title={A supply and demand framework for two-sided matching markets},
  author={Azevedo, Eduardo M and Leshno, Jacob D},
  journal={Journal of Political Economy},
  volume={124},
  number={5},
  pages={1235--1268},
  year={2016},
  publisher={University of Chicago Press Chicago, IL}
}

@article{ellickson1999clubs,
  title={Clubs and the Market},
  author={Ellickson, Bryan and Grodal, Birgit and Scotchmer, Suzanne and Zame, William R},
  journal={Econometrica},
  volume={67},
  number={5},
  pages={1185--1217},
  year={1999},
  publisher={Wiley Online Library}
}

@article{che2022prestige,
  title={Prestige seeking in college application and major choice},
  author={Che, Yeon-Koo and Hahm, Dong Woo and Kim, Jinwoo and Kim, Se-Jik and Tercieux, Olivier},
  journal={Available at SSRN 4309000},
  year={2025}
}

@article{phan2024crowding,
  title={Crowding in school choice},
  author={Phan, William and Tierney, Ryan and Zhou, Yu},
  journal={American Economic Review},
  volume={114},
  number={8},
  pages={2526--2552},
  year={2024},
  publisher={American Economic Association 2014 Broadway, Suite 305, Nashville, TN 37203}
}

@article{carmona2023existence,
  title={Existence of stable matchings in large economies with externalities},
  author={Carmona, Guilherme and Laohakunakorn, Krittanai},
  journal={University of Surrey.[1264]},
  year={2023}
}

@techreport{campos2024social,
  title={Social Interactions, Information, and Preferences for Schools: Experimental Evidence from Los Angeles},
  author={Campos, Christopher},
  year={2024},
  institution={National Bureau of Economic Research}
}

@article{beuermann2023good,
  title={What is a good school, and can parents tell? Evidence on the multidimensionality of school output},
  author={Beuermann, Diether W and Jackson, C Kirabo and Navarro-Sola, Laia and Pardo, Francisco},
  journal={The Review of Economic Studies},
  volume={90},
  number={1},
  pages={65--101},
  year={2023},
  publisher={Oxford University Press}
}

@online{joffe-walt2020introducing,
  author       = {Chana Joffe-Walt},
  title        = {Introducing ``Nice White Parents'''},
  year         = {2020},
  month        = {7},
  url          = {https://www.nytimes.com/2020/07/23/podcasts/nice-white-parents-serial.html},
  note         = {New York Times, podcasts section (accessed August 2025)}

}

@article{liu2016ordinal,
  title={Ordinal efficiency, fairness, and incentives in large markets},
  author={Liu, Qingmin and Pycia, Marek},
  journal={Working Paper},
  year={2016}
}

@article{abowd1999high,
  title={High Wage Workers and High Wage Firms},
  author={Abowd, John M. and Kramarz, Francis and Margolis, David N.},
  journal={Econometrica},
  volume={67},
  number={2},
  pages={251--333},
  year={1999},
  doi={10.1111/1468-0262.00020}
}

@article{bender2018management,
  title={Management Practices, Workforce Selection, and Productivity},
  author={Bender, Stefan and Bloom, Nicholas and Card, David and {Van Reenen}, John and Wolter, Stefanie},
  journal={Journal of Labor Economics},
  volume={36},
  number={S1},
  pages={S371--S409},
  year={2018},
  doi={10.1086/694107}
}

@article{nguyen2018near,
  title={Near-feasible stable matchings with couples},
  author={Nguyen, Thanh and Vohra, Rakesh},
  journal={American Economic Review},
  volume={108},
  number={11},
  pages={3154--3169},
  year={2018},
  publisher={American Economic Association 2014 Broadway, Suite 305, Nashville, TN 37203}
}

@article{diamond2016determinants,
  title={The determinants and welfare implications of US workers' diverging location choices by skill: 1980--2000},
  author={Diamond, Rebecca},
  journal={American economic review},
  volume={106},
  number={3},
  pages={479--524},
  year={2016},
  publisher={American Economic Association 2014 Broadway, Suite 305, Nashville, TN 37203}
}

@article{almagro2025location,
  title={Location sorting and endogenous amenities: Evidence from amsterdam},
  author={Almagro, Milena and Dom{\'\i}nguez-Iino, Tom{\'a}s},
  journal={Econometrica},
  volume={93},
  number={3},
  pages={1031--1071},
  year={2025},
  publisher={Wiley Online Library}
}

@article{bayer2007unified,
  title={A unified framework for measuring preferences for schools and neighborhoods},
  author={Bayer, Patrick and Ferreira, Fernando and McMillan, Robert},
  journal={Journal of political economy},
  volume={115},
  number={4},
  pages={588--638},
  year={2007},
  publisher={The University of Chicago Press}
}

\appendix

\newpage

\section{Additional Examples\label{sec:Additional-Examples}}

\paragraph{A Pareto-Dominated Stable Matching}
\begin{example}
\label{exa:Pareto-dominated} The economy is identical to the economy
of Example \ref{exa:multiple-matchings}, except that for a student
$\theta$ with $\chi^{\theta}\in\left[0,1/2\right]$ preference parameters
are $g_{c_{1}}^{\theta}=2,g_{c_{2}}^{\theta}=1$. For a student $\theta$
with $\chi^{\theta}\in\left[1/2,1\right]$ we maintain that $g_{c_{1}}^{\theta}=1,g_{c_{2}}^{\theta}=2$. 
\end{example}
The matchings $\mu^{\dagger},\mu^{\circ}$ given by 
\[
\begin{array}{ccc}
\mu^{\dagger}\left(\theta\right)=\begin{cases}
c_{1} & \chi^{\theta}\geq1/2\\
c_{2} & \chi^{\theta}<1/2
\end{cases}, &  & \mu^{\circ}\left(\theta\right)=\begin{cases}
c_{2} & \chi^{\theta}\geq1/2\\
c_{1} & \chi^{\theta}<1/2
\end{cases}\end{array}
\]
are stable matchings for the economy of Example \ref{exa:Pareto-dominated},
because this economy is identical to the economy of Example \ref{exa:multiple-matchings}
except for the change to the preferences of students who only have
a single college in their budget set. In the economy of Example \ref{exa:Pareto-dominated}
all students obtain higher utility under the matching $\mu^{\circ}$
than under $\mu^{\dagger}$, as students have identical peers under
both matchings but $\mu^{\circ}$ assigns each student $\theta$ to
the college that maximizes $g_{\mu\left(\theta\right)}^{\theta}$
and $\mu^{\dagger}$ assigns each student $\theta$ to the college
that minimizes $g_{\mu\left(\theta\right)}^{\theta}$.

\paragraph{Different Stable Matchings with Identical Cutoffs}
\begin{example}
\label{exa:Diff-Match-Same-Cutoffs}The economy is identical to the
economy of Example \ref{exa:multiple-matchings}, except that college
capacities are $q_{c_{1}}=q_{c_{2}}=2$, and the preferences parameters
of student $\theta$ with GPA $\chi^{\theta}$ are $\beta^{\theta}=16\cdot\left(\chi^{\theta}+1\right)$,
$g_{c_{1}}^{\theta}=1,g_{c_{2}}^{\theta}=2$.
\end{example}
In the economy of Example \ref{exa:Diff-Match-Same-Cutoffs}, if all students are assigned to college $c\in\left\{ c_{1},c_{2}\right\} $ then all students prefer college $c$. Thus, the matchings $\mu^{1}\left(\theta\right)\equiv c_{1}$
and $\mu^{2}\left(\theta\right)\equiv c_{2}$ are both stable matchings. The corresponding cutoffs and matching statistics are $P^{1}=\left(0,0\right),s^{1}=\left(\frac{1}{2},0\right)$ and $P^{2}=\left(0,0\right),s^{2}=\left(0,\frac{1}{2}\right)$. 

\vspace{1cm}

\paragraph{Failure of the Rural Hospital Theorem}
\begin{example}
\label{exa:No-rural-hopital}There are two colleges $\mathcal{C}=\left\{ c_{1},c_{2}\right\} $,
each with capacity $\frac{1}{3}$. The set of student types is $\Theta=\mathcal{X}\times\Gamma\times\left[0,1\right]^{2}$
with student attributes $\mathcal{X}=\left[0,1\right]$, and preference
parameters $\Gamma=\left[0,1\right]$. There are three groups of students.
Group 1 is a $1/3$ mass of students distributed uniformly over $G_{1}=\left\{ \theta\mid\chi^{\theta}=0,\gamma^{\theta}=r_{c_{1}}^{\theta}\in\left[\frac{1}{3},\frac{2}{3}\right],r_{c_{2}}^{\theta}=r_{c_{1}}^{\theta}-\frac{1}{3}\in\left[0,\frac{1}{3}\right]\right\} $,
that is, students in group 1 have attribute $\chi^{\theta}=0$, are
the middle priority group at $c_{1}$, and the lowest priority at
$c_{2}$. Group 2 is a $1/3$ mass of students distributed uniformly
over $G_{2}=\left\{ \theta\mid\chi^{\theta}=0,\gamma^{\theta}=r_{c_{1}}^{\theta}\in\left[0,\frac{1}{3}\right],r_{c_{2}}^{\theta}=r_{c_{1}}^{\theta}+\frac{1}{3}\in\left[\frac{1}{3},\frac{2}{3}\right]\right\} $,
that is, students in group 2 have attribute $\chi^{\theta}=0$, are
the middle priority group at $c_{2}$, and the lowest priority at
$c_{1}$. Group 3 is contains a $1/3$ mass of students which is uniformly
distributed over $G_{3}=\left\{ \theta\mid\chi^{\theta}=1,\gamma^{\theta}=r_{c_{1}}^{\theta}=r_{c_{2}}^{\theta}\in\left[\frac{2}{3},1\right]\right\} $,
that is, students in group 3 have attribute $\chi^{\theta}=1$ and
hold the top priorities at both colleges. The matching statistics
$s$ correspond to the mass of students from group 3 at each college.
Student preferences are given by
\[
u^{\theta}\left(c;s\right)=g_{c}^{\theta}+10\cdot s_{c},
\]
 where $g_{c_{1}}^{\theta}=1+\gamma^{\theta},g_{c_{2}}^{\theta}=1$
(the variation in $g_{c_{1}}^{\theta}-g_{c_{2}}^{\theta}$ ensures
that the diversity of preferences assumption is satisfied). 

In the economy of Example \ref{exa:No-rural-hopital} all students prefer $c_{1}\succ^{\theta|s}c_{2}$ if all students from group 3 are assigned to college $c_{1}$ (i.e., $s=\left(1/3,0\right)$), and all students have the reverse preferences $c_{2}\succ^{\theta|s}c_{1}$ if all students from group 3 are assigned to college $c_{2}$ (i.e., $s=\left(0,1/3\right)$). Thus, both of the following matchings are stable:
\[
\begin{array}{cc}
\mu^{1}\left(\theta\right)=\begin{cases}
\phi & \theta\in G_{1}\\
c_{2} & \theta\in G_{2}\\
c_{1} & \theta\in G_{3}
\end{cases},\,\,\,\,\, & \mu^{2}\left(\theta\right)=\begin{cases}
c_{1} & \theta\in G_{1}\\
\phi & \theta\in G_{2}\\
c_{2} & \theta\in G_{3}
\end{cases}\end{array}.
\]
 The stable matching $\mu^{1}$ assigns the students in $G_{2}$ but
leaves students in $G_{1}$ unassigned, while the stable matching
$\mu^{2}$ assigns the students in $G_{1}$ but leaves students in
$G_{2}$ unassigned.
\end{example}

\section{Omitted proofs}
\begin{proof}[Proof of Lemma \ref{lem:(cutoff<->match)}]
\textbf{\emph{}} Let the pair $\left(P,s\right)$ be rational expectations
and market clearing, and consider the matching $\mu$ defined by $\mu\left(\theta\right)=D^{\theta}\left(P,s\right)$.
Because $\left(P,s\right)$ is market clearing, for all $c\in\mathcal{C}$
we have that $\eta\left(\mu\left(c\right)\right)=D_{c}\left(P,s\right)\leq q_{c}$. 

Consider a student-college pair $\left(\theta,c\right)$ such that
$\theta$ prefers $c$ over $\mu\left(\theta\right)$ given the matching
$\mu$, i.e., $c\succ^{\theta|\mu}\mu\left(\theta\right)$. Because
\[
\mu\left(\theta\right)=D^{\theta}\left(P,s\right)=\arg\max_{\succ^{\theta|s}}B^{\theta}\left(P\right)
\]
and by rational expectations we have that $\succ^{\theta|\mu}=\succ^{\theta|s}$,
it must be that $c\notin B^{\theta}\left(P\right)$. That is, $r_{c}^{\theta}<P_{c}$.
This implies that $P_{c}>0$, and from market clearing it must be
that $\eta\left(\mu\left(c\right)\right)=D_{c}\left(P,s\right)=q_{c}$.
Any $\theta'\in\mu\left(c\right)$ must have higher priority than
$P_{c}$, that is, $r_{c}^{\theta'}\geq P_{c}>r_{c}^{\theta}$. Therefore,
$\left(\theta,c\right)$ is not a blocking pair. Since $\mu$ satisfies
capacity constraints and is not blocked, it is a stable matching. 

Conversely, let $\mu$ be a stable matching. Define $\left(P,s\right)$
by 
\begin{align*}
\begin{array}{cc}
P_{c} & =\begin{cases}
\inf\left\{ r_{c}^{\theta}\mid\theta\in\mu(c)\right\}  & \mbox{ if }\eta\left(\mu\left(c\right)\right)=q_{c}\\
0 & \mbox{ if }\eta\left(\mu\left(c\right)\right)<q_{c}
\end{cases},\end{array} &  & s=\sigma\left(\mu\right).
\end{align*}
For the sake of contradiction, suppose that for some student $\theta$
we have that $c=\mu\left(\theta\right)\neq D^{\theta}\left(P,s\right)=c'$.
Because $\theta\in\mu\left(c\right)$, we have that $r_{c}^{\theta}\geq P_{c}$
and $c\in B^{\theta}\left(P\right)$. Because $D^{\theta}\left(P,s\right)=\arg\max_{\succ^{\theta|s}}B^{\theta}\left(P\right)$
and by definition $\succ^{\theta|\mu}=\succ^{\theta|s}$, we have
that $c'\in B^{\theta}\left(P\right)$ and $c'\succ^{\theta|\mu}c$.
Thus, $r_{c'}^{\theta}\geq P_{c'}$, and either $\eta\left(\mu\left(c'\right)\right)<q_{c'}$
or there exists $\theta'\in\mu\left(c'\right)$ such that $r_{c'}^{\theta}\geq r_{c'}^{\theta'}$.
In the former case, $\left(\theta,c'\right)$ form a blocking pair.
In the latter case, by right continuity of $\mu$ there exists $\theta^{\circ}$
such that $c'\succ^{\theta^{\circ}|\mu}c$ and $r_{c'}^{\theta^{\circ}}>r_{c'}^{\theta'}$
and $\left(\theta^{\circ},c'\right)$ form a blocking pair. Either
contradicts the stability of $\mu$.

It thus follows that $\mu\left(\theta\right)=D^{\theta}\left(P,s\right)$
for all students $\theta$. Rational expectations immediately follows.
Market clearing follows because stability of $\mu$ implies that $D_{c}\left(P,s\right)=\eta\left(\mu\left(c\right)\right)\leq q_{c}$,
and for $D_{c}\left(P,s\right)=\eta\left(\mu\left(c\right)\right)<q_{c}$
we have by definition that $P_{c}=0$. 
\end{proof}

\subsubsection{Existence of a Stable Matching }
\begin{proof}[Proof of Theorem \ref{thm:(stable exists)}]
\textbf{\emph{}} Let $E=[\Theta,C,\eta,q]$ be an economy satisfying
Assumption \ref{ass:(Continuity with Stats)}. Recall that, without
loss of generality, we normalize the total measure so that $\eta\left(\Theta\right)=1$,
and normalize priorities so that a student $\theta$'s priority at
each college $c$ is the student's percentile in college $c$'s rankings,
that is, $r_{c}^{\theta}=\eta\left(\left\{ \theta'\mid r_{c}^{\theta}>r_{c}^{\theta'}\right\} \right)\in\left[0,1\right]$
for any $c\in\mathcal{C}$, $\theta\in\Theta$.

By Lemma \ref{lem:(cutoff<->match)}, it suffices to show the existence
of a rational expectations market-clearing cutoff $\left(P,s\right)$.
Let $Z\left(P,s\right)=Z^{P}\left(P,s\right)\times Z^{s}\left(P,s\right)$
be defined by 

\begin{eqnarray*}
Z_{c}^{P}\left(P,s\right) & = & \begin{cases}
\frac{P_{c}}{1+q_{c}-D_{c}\left(P,s\right)} & \mbox{if }D_{c}(p,s)\leq q_{c}\\
P_{c}+D_{c}\left(P,s\right)-q_{c} & \mbox{if }D_{c}(p,s)\geq q_{c}
\end{cases}
\end{eqnarray*}
and 

\[
Z^{s}\left(P,s\right)=\sigma\left(\mu\right)\mbox{ for }\mu\text{ defined by }\mu\left(\theta\right)=D^{\theta}\left(P,s\right)\,.
\]

First, observe that if $\left(P,s\right)$ is a fixed point of $Z$,
then $\left(P,s\right)$ is rational expectations and market clearing.
By definition, $Z^{s}\left(P,s\right)=s$ implies that $\left(P,s\right)$
satisfies the rational expectations condition (ii). If $Z^{p}\left(P,s\right)=P$,
then for every college $c$ either (a) $D_{c}\left(P,s\right)=q_{c}$,
or (b) $D_{c}\left(P,s\right)\leq q_{c}$ and $P_{c}=0$. Therefore,
$Z^{p}\left(P,s\right)=P$ implies that $\left(P,s\right)$ satisfies
the market clearing condition (i). 

Second, we show that $Z$ is a continuous function from the compact,
convex set $\left[0,1\right]^{\mathcal{C}}\times\mathcal{X^{C}}$
to itself. Because $\mathcal{X}$ is convex, closed, and $0\in\mathcal{X}$,
we have that $Z^{s}\left(P,s\right)\in\mathcal{X^{C}}$. To see that
$Z^{P}\left(P,s\right)\in[0,1]^{\mathcal{C}}$ observe that if $D_{c}\left(P,s\right)\leq q_{c}$,
we have $0\leq Z_{c}^{P}\left(P,s\right)\leq P_{c}\leq1$, and if
$D_{c}\left(P,s\right)\geq q_{c}$, we have $0\leq Z_{c}^{p}\left(P,s\right)\leq1-q_{c}\leq1$
because $D_{c}\left(P,s\right)\leq\eta\left(\left\{ \theta\mid r_{c}^{\theta}\geq P_{c}\right\} \right)=1-P_{c}$.
Therefore, $Z$ is a function from $\left[0,1\right]^{\mathcal{C}}\times\mathcal{X^{C}}$
to itself.

To show that $Z$ is a continuous function, we show that $Z^{s}\left(\cdot,\cdot\right)$
and $D\left(\cdot,\cdot\right)$ are continuous. For $\left(P,s\right),\left(P',s'\right)$
define 
\begin{align*}
\Delta\left(P,s;P',s'\right)= & \left\{ \theta\mid B^{\theta}\left(P\right)\neq B^{\theta}\left(P'\right)\right\} \cup\left\{ \theta\mid\succ^{\theta|s}\neq\succ^{\theta|s'}\right\} 
\end{align*}
to be the set of students who have a different budget set under $P$
and $P'$ or have different preferences under $s$ and $s'$. For
any $\theta\notin\Delta\left(P,s;P',s'\right)$ we have that $D^{\theta}\left(P,s\right)=D^{\theta}\left(P',s'\right)$,
and thus
\begin{align}
\left\Vert D\left(P,s\right)-D\left(P',s'\right)\right\Vert _{\infty} & \leq\eta\left(\Delta\left(P,s;P',s'\right)\right)\label{eq:bound-D-diff}\\
\left\Vert Z^{s}\left(P,s\right)-Z^{s}\left(P',s'\right)\right\Vert _{\infty} & \leq\sup_{x\in\mathcal{X}}\left\Vert x\right\Vert _{\infty}\cdot\eta\left(\Delta\left(P,s;P',s'\right)\right).\nonumber 
\end{align}
Denote $M=\sup_{x\in\mathcal{X}}\left\Vert x\right\Vert _{\infty}$.
Because $\mathcal{X}$ is bounded, we have that $M<\infty$. We have
that\footnote{We use $\oplus$ to denote the symmetric difference between sets,
that is $A\oplus B=A\cup B\setminus A\cap B$. } {\small{}
\begin{align*}
\eta\left(\left\{ \theta\mid B^{\theta}\left(P\right)\neq B^{\theta}\left(P'\right)\right\} \right) & \leq\sum_{c}\eta\left(\left\{ \theta\mid c\in B^{\theta}\left(P\right)\oplus B^{\theta}\left(P'\right)\right\} \right)\\
 & =\sum_{c}\eta\left(\left\{ \theta\mid r_{c}^{\theta}\in\left[\min\left(P_{c},P'_{c}\right),\max\left(P_{c},P'_{c}\right)\right]\right\} \right)\\
 & =\sum_{c}\left|P_{c}-P'_{c}\right|\\
 & =\left\Vert P-P'\right\Vert _{1}\\
 & \leq J\left\Vert P-P'\right\Vert _{\infty}\,,
\end{align*}
}where the equality in the third line follows from our normalization of $r_{c}^{\theta}$ (recall that $J=\left|\mathcal{C}\right|$ is the number of colleges). 

By Assumption \ref{ass:(Continuity with Stats)}, for any $\varepsilon>0$
there exists $\delta>0$ such that for any $\left\Vert s-s'\right\Vert _{\infty}<\delta$
we have that $\eta\left(\left\{ \theta\,\mid\,\succ^{\theta|s}\neq\succ^{\theta|s'}\right\} \right)<\frac{\varepsilon}{2M}$.
Together, we have that if $\left\Vert P-P'\right\Vert _{\infty},\left\Vert s-s'\right\Vert _{\infty}<\min\left(\delta,\frac{\varepsilon}{2MJ}\right)$, then we have that 
\begin{align*}
\eta\left(\Delta\left(P,s;P',s'\right)\right) & \leq\eta\left(\left\{ \theta\mid B^{\theta}\left(P\right)\neq B^{\theta}\left(P'\right)\right\} \right)+\eta\left(\left\{ \theta\,\mid\,\succ^{\theta|s}\neq\succ^{\theta|s'}\right\} \right)\\
 & \leq\frac{\varepsilon}{2M}+\frac{\varepsilon}{2M}\\
 & =\varepsilon/M\,.
\end{align*}
Plugging this bound into equation (\ref{eq:bound-D-diff}) shows that
$Z^{s}\left(\cdot,\cdot\right)$ and $D\left(\cdot,\cdot\right)$
are continuous in the sup-norm, and thus $Z$ is a continuous
function. 

By Brouwer's fixed-point theorem, $Z$ has a fixed point $Z\left(P^{*},s^{*}\right)=\left(P^{*},s^{*}\right)$.
We have that $\left(P^{*},s^{*}\right)$ is rational expectation and
market clearing. By Lemma \ref{lem:(cutoff<->match)} $\mu^{*}$ defined
by $\mu\left(\theta\right)=D^{\theta}\left(P^{*},s^{*}\right)$ is
a stable matching.
\end{proof}

\subsubsection{Existence of an Approximately Stable Matching in Large Sampled Economies}
\begin{proof}[Proof of Theorem \ref{thm:aprox-stable-for-sampled}]

The proof builds upon the convergence of standard economies (see the
online appendix of \citet{azevedo2016supply}) to show that with high
probability we can find $\left(P^{n},s^{n}\right)$ that are close
to $\left(P^{*},s^{*}\right)$ such that $\mu^{n}\left(\theta\right)=D^{\theta}\left(P^{n},s^{n}\right)$
is an $\varepsilon$-stable matching. For ease of reading, steps in the proof are stated as claims. 

Fixing $s^{*}$, we generate a standard continuum economy $\tilde{E}=[\tilde{\Theta},\mathcal{C},\tilde{\eta},q]$
in which $\tilde{\theta}=\left(\succ^{\theta},r^{\theta}\right)\in\tilde{\Theta}$
and $\tilde{\eta}$ is given by 
\[
\tilde{\eta}\left(\left\{ \tilde{\theta}\mid\succ^{\theta}\in O,r^{\theta}\in R\right\} \right)=\eta\left(\left\{ \theta\mid\succ^{\theta|s^{*}}\in O,r^{\theta}\in R\right\} \right)
\]
for any subset of preference orderings $O$ and a subset of priorities
$R\subset\left[0,1\right]^{\mathcal{C}}$. Likewise, by fixing $s^{*}$
we reduce a sampled economy $F^{n}$ to an economy $\tilde{F}^{n}$
without peer-dependent preferences. Because $\tilde{\eta}$ is a projection
of $\eta$, the distribution of $\tilde{F}^{n}$ is identical to the
distribution of economies sampled from $\tilde{\eta}$. Let $D\left(P\mid\tilde{\eta}\right)=D\left(P,s^{*}\mid\eta\right)$
denote the demand for economy $\tilde{E}$.

Cutoffs $P^{*}$ are market clearing for $\tilde{E}=[\tilde{\Theta},\mathcal{C},\tilde{\eta},q]$.
By applying the rural hospital theorem to the standard economy $\tilde{E}$,
if $D_{c}\left(P^{*}\mid\tilde{\eta}\right)<q_{c}$, college $c$
does not fill its capacity under any stable matching and $P_{c}=0$
for any cutoffs $P$ that are market clearing for $\tilde{E}$. Let
$C\left(P^{*}\right)=\left\{ c\in\mathcal{C}\mid D_{c}\left(P^{*}\mid\tilde{\eta}\right)=q_{c}\right\} $
denote the set of colleges that can have a strictly positive market
clearing cutoff, and let $\mathcal{N}\left(P^{*}\right)=\left\{ P\in\left[0,1\right]^{\mathcal{C}}\mid P_{c}=0\text{ for all }c\notin C\left(P^{*}\right)\right\} $
denote the set of cutoffs that differ from $P^{*}$ only for colleges
in $C\left(P^{*}\right)$. Let $\mathcal{B}_{\varepsilon}\left(P^{*}\right)=\left\{ P\in\left[0,1\right]^{\mathcal{C}}\mid\left\Vert P-P^{*}\right\Vert _{\infty}<\varepsilon\right\} $
denote the $\varepsilon$-ball around $P^{*}$.

Define\footnote{We use the notation $\left[x\right]^{+}=\max\left\{ x,0\right\} $.}
\[
\mathcal{Z}\left(P\mid q,\tilde{\eta}\right)=\max_{c\in C\left(P^{*}\right)}\left\{ P_c\cdot\left[q_{c}-D_{c}\left(P\mid\tilde{\eta}\right)\right]^{+}+\left[D_{c}\left(P\mid\tilde{\eta}\right)-q_{c}\right]^{+}\right\} \,,
\]
 and observe that $\mathcal{Z}\left(P\mid q,\tilde{\eta}\right)\geq0$
for any $P$, and $\mathcal{Z}\left(P\mid q,\tilde{\eta}\right)=0$
if and only if $P$ is market clearing for $\tilde{E}=[\tilde{\Theta},\mathcal{C},\tilde{\eta},q]$. 

\paragraph*{Claim 1: }

For any $\varepsilon_{1}>0$ there exists $m_{1}>0$ such that for
any $P\in\mathcal{N}\left(P^{*}\right)\setminus\mathcal{B}_{\varepsilon_{1}}\left(P^{*}\right)$
we have that $\mathcal{Z}\left(P\mid q,\tilde{\eta}\right)\geq m_{1}\,$.

\,

Fix $\varepsilon_{1}>0$. Observe that $D\left(P\mid\tilde{\eta}\right)$
is continuous in $P$ because 
\begin{align*}
\left\Vert D\left(P\mid\tilde{\eta}\right)-D\left(P'\mid\tilde{\eta}\right)\right\Vert _{\infty} & =\left\Vert D\left(P,s^{*}\mid\eta\right)-D\left(P',s^{*}\mid\eta\right)\right\Vert _{\infty}\\
 & \leq\eta\left(\left\{ \theta\mid B^{\theta}\left(P\right)\neq B^{\theta}\left(P'\right)\right\} \right)\\
 & \leq\left\Vert P-P'\right\Vert _{1}\,.
\end{align*}
Assume, for the sake of contradiction, that no such $m_{1}>0$ exists.
From continuity of $D\left(\cdot\mid\tilde{\eta}\right)$, the function
$\mathcal{Z}\left(P\right)$ is continuous, and there exists $P^{\dagger}\in\mathcal{N}\left(P^{*}\right)\setminus\mathcal{B}_{\varepsilon_{1}}\left(P^{*}\right)$
such that $\mathcal{Z}\left(P^{\dagger}\mid q,\tilde{\eta}\right)=0$.

We construct a standard economy $\tilde{E}^{\dagger}$ in which $P^{*},P^{\dagger}$
are market-clearing. Construct the economy $\tilde{E}^{\dagger}=[\tilde{\Theta},\mathcal{C},\tilde{\eta},q^{\dagger}]$
to be identical to $\tilde{E}$ except for having college capacities
$q_{c}^{\dagger}=\tilde{\eta}\left(\tilde{\Theta}\right)$ for $c\notin C\left(P^{*}\right)$,
and $q_{c}^{\dagger}=q_{c}$ for $c\in C\left(P^{*}\right)$. That
is, we increase the capacity of colleges that were under-demanded
under $P^{*}$ to have sufficient capacity to admit all students.
For $c\notin C\left(P^{*}\right)$ we have $P_{c}^{\dagger}=P_{c}^{*}=0$
and $D_{c}\left(P^{*}\mid\tilde{\eta}\right),D_{c}\left(P^{\dagger}\mid\tilde{\eta}\right)\leq\tilde{\eta}\left(\tilde{\Theta}\right)=q_{c}^{\dagger}$.
For all $c\in C\left(P^{*}\right)$ we have that $D_{c}\left(P^{*}\mid\tilde{\eta}\right)=q_{c}=q_{c}^{\dagger}$,
and $D_{c}\left(P^{\dagger}\mid\tilde{\eta}\right)\leq q_{c}=q_{c}^{\dagger}$
with $P^{\dagger}=0$ if $D_{c}\left(P^{\dagger}\mid\tilde{\eta}\right)<q_{c}$
because $\mathcal{Z}\left(P^{\dagger}\mid q,\tilde{\eta}\right)=0$.
Thus, $P^{*},P^{\dagger}$ are market clearing for $\tilde{E}^{\dagger}$. 

By the lattice structure, $P^{+}=\max\left\{ P^{*},P^\dagger\right\} $ and
$P^{-}=\min\left\{ P^{*},P^\dagger\right\} $ are also market clearing cutoffs
for $\tilde{E}^{\dagger}$. Without loss, suppose that $P^{+}\neq P^{*}$.
By the lattice structure of market-clearing cutoffs and the rural
hospital theorem \citep{azevedo2016supply}, $D\left(P^{*}\mid\tilde{\eta}\right)=D\left(P\mid\tilde{\eta}\right)$
for any $P\in\text{Conv}\left(P^{*},P^{+}\right)$. This implies that
if $\partial_{P}D\left(P^{*},s^{*}\right)$ exists, it must be singular.
This contradicts the theorem's assumption, thereby proving the claim.

\paragraph*{Claim 2:}

For any $\varepsilon>0$ and $m>0$, there exists $N$ such that for
any $n>N$ with probability of at least $1-\varepsilon$, demand in
the sampled economy $\tilde{F}^{n}$ satisfies 
\[
\sup_{P\in\left[0,1\right]^{\mathcal{C}}}\bigl\Vert D\left(P\mid\tilde{\eta}^{n}\right)-D\left(P\mid\tilde{\eta}\right)\bigr\Vert_{\infty}\leq m\,.
\]

\,

The claim follows from the Vapnik-Chervonenkis Theorem \citep{vapnik1972uniform}
and the fact that the $n$-th shatter coefficient (which is the maximum
number of possible ways $n$ students can be assigned to $J$ colleges
by some cutoffs $P$) is bounded by $a\cdot n^{J}$ for some $a>0$.
Namely, there exists $a>0$ such that the probability that demand
deviates by more than $m$ for some $P$ is bounded by
\begin{equation}
\Pr\{\sup_{P\in\left[0,1\right]^{\mathcal{C}}}\left\Vert D\left(P\mid\tilde{\eta}^{n}\right)-D\left(P\mid\tilde{\eta}\right)\right\Vert _{\infty}>m\}\leq a\cdot n^{J}\cdot e^{-nm^{2}/32}.\label{eq:bound-diff-demand}
\end{equation}
Because $a\cdot n^{J}\cdot e^{-nm^{2}/32}\underset{n\to\infty}{\rightarrow}0$, there exists $N$ such that this probability is less than $\varepsilon$ for any $n>N$.

\paragraph*{Claim 3:}

For any $\varepsilon_{1},\varepsilon_{2}>0$, there is $N$ such that
for any $n>N$ with probability of at least $1-\varepsilon_{2}$ the
sampled economy $\tilde{F}^{n}$ has a market clearing cutoff in $\mathcal{N}\left(P^{*}\right)\cap\mathcal{B}_{\varepsilon_{1}}\left(P^{*}\right)$
.

\,

Let there be $\varepsilon_{1},\varepsilon_{2}>0$. By Claim 1, there
exists $m_{1}>0$ such that for any $P\in\mathcal{N}\left(P^{*}\right)\setminus\mathcal{B}_{\varepsilon_{1}}\left(P^{*}\right)$
we have that $\mathcal{Z}\left(P\mid q,\tilde{\eta}\right)\geq m_{1}$.
If $C\left(P^{*}\right)\neq\mathcal{C}$, let $m_{2}=\min_{c\notin C\left(P^{*}\right)}\left\{ q_{c}-D_{c}\left(P^{*}\mid\tilde{\eta}\right)\right\} $
and define $m=\min\left\{ m_{1}/4,m_{2}/4J\right\} $. If $C\left(P^{*}\right)=\mathcal{C}$,
define $m=m_{1}/4$. 

By Claim 2, there exists $N$ such that for all $n>N$ we have that
with probability of at least $1-\varepsilon_{2}$
\begin{align}
\label{eq:bound_from_2}
\left|D_{c}\left(P\mid\tilde{\eta}^{n}\right)-D_{c}\left(P\mid\tilde{\eta}\right)\right|\leq m   
\end{align}
holds for all $P\in\left[0,1\right]^{\mathcal{C}}$ and $c\in\mathcal{C}$.
In addition, we take $N>1/m$ so that $\left\Vert q-q^{n}\right\Vert \leq1/n<m$
for any $n>N$. Let $\tilde{F}^{n}$ be such an economy. Because $\tilde{F}^{n}$ is a standard economy, it has a stable matching with associated market-clearing cutoffs $P^{n}$. We will show that $P^{n}\in\mathcal{N}\left(P^{*}\right)\cap\mathcal{B}_{\varepsilon_{1}}\left(P^{*}\right)$.

When $C\left(P^{*}\right)=\mathcal{C},$ $P^{n}\in\mathcal{N}\left(P^{*}\right)$
trivially holds. To see that $P^{n}\in\mathcal{N}\left(P^{*}\right)$
when $C\left(P^{*}\right)\neq\mathcal{C}$ we calculate the cutoffs
for an economy with weakly higher demand than $\tilde{F}^{n}$ through
vacancy chains (\citealt{blum1997vacancy}, \citealt{dworczak2016deferred})
and show that these vacancy chains cannot fill the capacity of a college
$c\notin C\left(P^{*}\right)$.\footnote{While $\tilde{F}^{n}$ is a discrete economy, we can construct an
economy with non-atomic students by replacing each atom $\theta=\left(\succ^{\theta},r^{\theta}\right)$
of $\tilde{\eta}^{n}$ with a uniform positive density over the set
of student types $\left\{ \left(\succ,r\right)\mid\succ=\succ^{\theta},\,r_{c}\in\left[r_{c}^{\theta}-\delta,r_{c}^{\theta}\right]\right\} $
for sufficiently small $\delta$. Because all college capacities in
$\tilde{F}^{n}$ are multiples of $1/n$, cutoffs never split the
mass corresponding to a single student. } 

By inequality \ref{eq:bound_from_2}, demand under $\tilde{\eta}^{n}$ is uniformly bounded by the demand
from $\tilde{\eta}$ plus the demand from some mass of $m\cdot J\leq m_{2}/4$
additional students $\tilde{\eta}'$.\footnote{For example, construct such $\tilde{\eta}'$ by taking for each $c\in\mathcal{C}$ a mass $m$ of top-ranked students who most prefer $c$. } Consider the standard economy $E^{\circ}=[\tilde{\Theta},\mathcal{C},\tilde{\eta}+\tilde{\eta}',q^{n}]$.
We can find a market clearing cutoffs for $E^{\circ}$ by starting from having the mass of students $\tilde{\eta}$ assigned according to $P^{*}$ and adjust cutoffs by (i) starting rejection chains by having all the colleges for which $q_{c}^{n}<q_{c}$ reject their lowest ranked students and letting students apply until rejection chains are resolved, and then (ii) starting rejection chains by having all the students in $\tilde{\eta}'$ apply. Because the rejection chains in (i) and (ii) can each add at most $m_{2}/4$ to the mass assigned to any college, these rejection chains produce market clearing cutoffs for $E^{\circ}$ under which any college $c\notin C\left(P^{*}\right)$ did not fill its capacity. We can also find market clearing cutoffs for $E^{\circ}$ by starting from having the mass $\tilde{\eta}^{n}$ assigned according to their demand under $P^{n}$, and starting rejection chains from the additional demand from $\tilde{\eta}+\tilde{\eta}'$. 
The resulting market clearing cutoffs for $E^{\circ}$ must be weakly larger than $P^{n}$, and equal to $0$ for any $c\notin C\left(P^{*}\right)$ by the rural hospital theorem. Therefore, $P^{n}\in\mathcal{N}\left(P^{*}\right)$.

To see that $P^{n}\in\mathcal{N}\left(P^{*}\right)\cap\mathcal{B}_{\varepsilon_{1}}\left(P^{*}\right)$,
observe that for any $P\in\mathcal{N}\left(P^{*}\right)\setminus\mathcal{B}_{\varepsilon_{1}}\left(P^{*}\right)$
we have that {\footnotesize{}
\begin{align*}
\mathcal{Z}\left(P\mid q^{n},\tilde{\eta}^{n}\right) & =\max_{c\in C\left(P^{*}\right)}\left\{ P\cdot\left[q_{c}^{n}-D_{c}\left(P\mid\tilde{\eta}^{n}\right)\right]^{+}+\left[D_{c}\left(P\mid\tilde{\eta}^{n}\right)-q_{c}^{n}\right]^{+}\right\} \\
 & \geq\max_{c\in C\left(P^{*}\right)}\left\{ P\cdot\left[q_{c}-D_{c}\left(P\mid\tilde{\eta}\right)\right]^{+}+\left[D_{c}\left(P\mid\tilde{\eta}\right)-q_{c}\right]^{+}\right\} -\frac{1}{n}-m\\
 & =\mathcal{Z}\left(P\mid q,\tilde{\eta}\right)-\frac{1}{n}-m\\
 & \geq m_{1}/2>0\,,
\end{align*}
}and thus $P^{n}$ is not market clearing. Therefore, it must be that
$P^{n}\in\mathcal{N}\left(P^{*}\right)\cap\mathcal{B}_{\varepsilon_{1}}\left(P^{*}\right)$.

\paragraph*{Claim 4: }

For any $\varepsilon>0$, there is $N$ such that for any $n>N$ with
probability of at least $1-\varepsilon$ the sampled economy $F^{n}$
has an $\varepsilon$-stable matching. 

\,

Let there be $\varepsilon>0$. Denote $M=\sup_{x\in\mathcal{X}}\left\Vert x\right\Vert _{\infty}$,
and let $\varepsilon_{1}=\frac{\varepsilon}{2\alpha\left(6MJ+1\right)}>0$.
By the central limit theorem, for every $c\in\mathcal{C}$ 
\[
\int_{\theta\in\mu^{*}\left(c\right)}\chi^{\theta}d\eta^{n}\overset{d}{\to}\int_{\theta\in\mu^{*}\left(c\right)}\chi^{\theta}d\eta
\]
\[
\eta^{n}\left(\left\{ \theta\mid r_{c}^{\theta}\in\left[P_{c}^{*}-\varepsilon_{1},P_{c}^{*}+\varepsilon_{1}\right]\right\} \right)\overset{d}{\to}\min\left\{2\varepsilon_{1},P_{c}^{*}+\varepsilon_{1}\right\}\,.
\]
Therefore, there exists $N_{1}$ such that for every $n>N_{1}$ with
probability of at least $1-\varepsilon/2$ we have that for economy
$F^{n}$ it holds that 
\begin{equation}
\left|\int_{\theta\in\mu^{*}\left(c\right)}\chi^{\theta}d\eta^{n}-\int_{\theta\in\mu^{*}\left(c\right)}\chi^{\theta}d\eta\right|\leq\varepsilon_{1}\label{eq:bound-diver-stats}
\end{equation}
and for all $c\in\mathcal{C}$
\begin{equation}
\eta^{n}\left(\left\{ \theta\mid r_{c}^{\theta}\in\left[P_{c}^{*}-\varepsilon_{1},P_{c}^{*}+\varepsilon_{1}\right]\right\} \right)\leq3\varepsilon_{1}\,.\label{eq:bound-diver-cutoffs}
\end{equation}
By Claim 3, for any $\varepsilon_{1}>0$ there exists $N_{2}$ such
that for any $n>N_{2}$ with probability of at least $1-\varepsilon/2$
the sampled economy $\tilde{F}^{n}$ has a market clearing cutoff
$P^{n}$ for $F^{n}$ such that $P^{n}\in\mathcal{N}\left(P^{*}\right)\cap\mathcal{B}_{\varepsilon_{1}}\left(P^{*}\right)$.\\
\indent Therefore, for $n>N=\max\left\{ N_{1},N_{2}\right\} $ with
probability of at least $1-\varepsilon$ the economy $F^{n}$ satisfies
(\ref{eq:bound-diver-stats}), (\ref{eq:bound-diver-cutoffs}) for
all $c\in\mathcal{C}$, and has a market clearing cutoff $P^{n}$
such that $P^{n}\in\mathcal{N}\left(P^{*}\right)\cap\mathcal{B}_{\varepsilon_{1}}\left(P^{*}\right)$.
We show that in such economies the matching $\mu$ defined by $\mu^{n}\left(\theta\right)=D^{\theta}\left(P^{n},s^{*}\right)$
is $\varepsilon$-stable. Let $\mu^{*}$ for $E$ be defined by $\mu^{*}\left(\theta\right)=D^{\theta}\left(P^{*},s^{*}\right)$
and let $s^{n}=\sigma\left(\mu^{n}\mid\eta^{n}\right)$.\\
\\
\indent We have that {\footnotesize{}
\begin{align*}
\left|s_{c}^{n}-s_{c}^{*}\right| & =\left|\sigma_{c}\left(\mu^{n}\mid\eta^{n}\right)-\sigma_{c}\left(\mu^{*}\mid\eta\right)\right|\\
 & =\left|\int_{\theta\in\mu^{n}\left(c\right)}\chi^{\theta}d\eta^{n}-\int_{\theta\in\mu^{*}\left(c\right)}\chi^{\theta}d\eta\right|\\
 & \leq\left|\int_{\theta\in\mu^{n}\left(c\right)}\chi^{\theta}d\eta^{n}-\int_{\theta\in\mu^{*}\left(c\right)}\chi^{\theta}d\eta^{n}\right|+\left|\int_{\theta\in\mu^{*}\left(c\right)}\chi^{\theta}d\eta^{n}-\int_{\theta\in\mu^{*}\left(c\right)}\chi^{\theta}d\eta\right|\\
 & \leq2M\eta^{n}\left(\mu^{n}\left(c\right)\oplus\mu^{*}\left(c\right)\right)+\varepsilon_{1}\\
 & \leq\left(6MJ+1\right)\varepsilon_{1}\\
 & =\frac{\varepsilon}{2\alpha}\,,
\end{align*}
}where the last inequality follows because {\footnotesize{}
\begin{align*}
\eta^{n}\left(\mu^{n}\left(c\right)\oplus\mu^{*}\left(c\right)\right) & \leq\eta^{n}\left(\left\{ \theta\mid\mu^{n}\left(\theta\right)\neq\mu^{*}\left(\theta\right)\right\} \right)\\
 & =\eta^{n}\left(\left\{ \theta\mid D^\theta\left(P^{n},s^{*}\mid \eta^n\right)\neq D^{\theta}\left(P^{*},s^{*}\mid \eta \right)\right\} \right)\\
 & \leq\sum_{c}\eta^{n}\left(\left\{ \theta\mid r_{c}^{\theta}\in\left[\min\left(P_{c}^{*},P{}_{c}^{n}\right),\max\left(P_{c}^{*},P{}_{c}^{n}\right)\right]\right\} \right)\\
 & \leq\sum_{c}\eta^{n}\left(\left\{ \theta\mid r_{c}^{\theta}\in\left[P_{c}^{*}-\varepsilon_{1},P_{c}^{*}+\varepsilon_{1}\right]\right\} \right)\\
 & \leq J\cdot3\varepsilon_{1}\,.
\end{align*}
}To see that the matching $\mu$ is $\varepsilon$-stable, consider
a student $\theta$ who can form a blocking pair with college $c$.
By definition, $c\in B^{\theta}\left(P^{n}\right)$ and we have that{\footnotesize{}
\begin{align*}
u^{\theta}\left(\mu^{n}\left(\theta\right);s^{n}\right) & \geq u^{\theta}\left(\mu^{n}\left(\theta\right);s^{*}\right)-\alpha\left\Vert s^{n}-s^{*}\right\Vert _{\infty}\\
 & \geq u^{\theta}\left(\mu^{n}\left(\theta\right);s^{*}\right)-\varepsilon/2\\
 & \geq u^{\theta}\left(c;s^{*}\right)-\varepsilon/2\\
 & \geq u^{\theta}\left(c;s^{n}\right)-\alpha\left\Vert s^{n}-s^{*}\right\Vert _{\infty}-\varepsilon/2\\
 & \geq u^{\theta}\left(c;s^{n}\right)-\varepsilon
\end{align*}
}where the first and fourth inequalities follow because utilities
are $\alpha$ peer-smooth at $s^{*}$, and the third inequality follows
because $u^{\theta}\left(\mu\left(\theta\right);s^{*}\right)=u^{\theta}\left(D^{\theta}\left(P^{n},s^{*}\right)\right)=\max_{c\in B^{\theta}\left(P^{n}\right)}\left\{ u^{\theta}\left(c;s^{*}\right)\right\} $.
\end{proof}
\vspace{-1cm}

\subsubsection*{Additional Proofs}
\begin{proof}[Proof of Lemma (\ref{lem:linear-economy-is-smooth})]
 Suppose that $F$ has a bounded density $f$ with $\left|f\left(x\right)\right|\leq\bar{f}$
for all $x\in\mathbb{R}$. Let there be $\varepsilon>0$. There exists
$\bar{\beta}$ such that $\eta\left(\left\{ \theta\mid\left\Vert \beta^{\theta}\right\Vert _{\infty}>\bar{\beta}\right\} \right)<\varepsilon/2$.
We have that if $\left\Vert s-s'\right\Vert _{\infty}<\varepsilon/\left(4J^{2}L\bar{\beta}\bar{f}\right)$
then{\footnotesize{}
\begin{align*}
\eta\left(\left\{ \theta\mid\succ^{\theta|s}\neq\succ^{\theta|s'}\right\} \right) & \leq\varepsilon/2+\eta\left(\left\{ \theta\mid\left\Vert \beta^{\theta}\right\Vert _{\infty}\leq\bar{\beta},\,\succ^{\theta|s}\neq\succ^{\theta|s'}\right\} \right)\\
 & \leq\varepsilon/2+\sum_{c_{1},c_{2}\in\mathcal{C}}\eta\left(\left\{ \theta\left|\left\Vert \beta^{\theta}\right\Vert _{\infty}\leq\bar{\beta},\begin{array}{c}
u^{\theta}\left(c_{1};s\right)<u^{\theta}\left(c_{2};s\right)\\
u^{\theta}\left(c_{1};s'\right)>u^{\theta}\left(c_{2};s'\right)
\end{array}\right.\right\} \right)\\
 & =\varepsilon/2+\sum_{c_{1},c_{2}\in\mathcal{C}}\eta\left(\left\{ \theta\left|\left\Vert \beta^{\theta}\right\Vert _{\infty}\leq\bar{\beta},\begin{array}{c}
g_{c_{1}}^{\theta}+\beta^{\theta}\cdot s_{c_{1}}+\xi_{c_{1}}^{\theta}<g_{c2}^{\theta}+\beta^{\theta}\cdot s_{c_{2}}+\xi_{c_{2}}^{\theta}\\
g_{c_{1}}^{\theta}+\beta^{\theta}\cdot s'_{c_{1}}+\xi_{c_{1}}^{\theta}>g_{c2}^{\theta}+\beta^{\theta}\cdot s'_{c_{2}}+\xi_{c_{2}}^{\theta}
\end{array}\right.\right\} \right)\\
 & =\varepsilon/2+\sum_{c_{1},c_{2}\in\mathcal{C}}\eta\left(\left\{ \theta\left|\left\Vert \beta^{\theta}\right\Vert _{\infty}\leq\bar{\beta},\begin{array}{c}
\xi_{c_{1}}^{\theta}<g_{c2}^{\theta}-g_{c_{1}}^{\theta}+\xi_{c_{2}}^{\theta}+\beta^{\theta}\cdot\left(s_{c_{2}}-s_{c_{1}}\right)\\
\xi_{c_{1}}^{\theta}>g_{c2}^{\theta}-g_{c_{1}}^{\theta}+\xi_{c_{2}}^{\theta}+\beta^{\theta}\cdot\left(s'_{c_{2}}-s'_{c_{1}}\right)
\end{array}\right.\right\} \right)\\
 & \leq\varepsilon/2+\sum_{c_{1},c_{2}\in\mathcal{C}}\bar{f}\sup_{\left\Vert \beta^{\theta}\right\Vert _{\infty}\leq\bar{\beta}}\left(\beta^{\theta}\cdot\left(s_{c_{2}}-s_{c_{1}}\right)-\beta^{\theta}\cdot\left(s'_{c_{2}}-s'_{c_{1}}\right)\right)\\
 & \leq\varepsilon/2+2J^{2}L\bar{\beta}\left\Vert s-s'\right\Vert _{\infty}\bar{f}\\
 & \leq\varepsilon\,,
\end{align*}
}{\footnotesize\par}

and thus $E$ satisfies Assumption (\ref{ass:(Continuity with Stats)}).

For the second part of the lemma, suppose that $\beta_{\max}=\sup_{\theta\in\Theta}\left\Vert \beta^{\theta}\right\Vert _{\infty}<\infty$.
Then 
\begin{align*}
\left|u^{\theta}\left(c;s\right)-u^{\theta}\left(c;s'\right)\right| & =\left|\left(g_{c_{1}}^{\theta}+\beta^{\theta}\cdot s_{c_{1}}+\xi_{c_{1}}^{\theta}\right)-\left(g_{c_{1}}^{\theta}+\beta^{\theta}\cdot s'_{c_{1}}+\xi_{c_{1}}^{\theta}\right)\right|\\
 & =\left|\beta^{\theta}\cdot\left(s_{c_{1}}-s'_{c_{1}}\right)\right|\\
 & \leq L\beta_{\max}\left\Vert s-s'\right\Vert _{\infty}
\end{align*}
and we have that utilities are $L\beta_{\max}$ peer-smooth. 
\end{proof}

\section{Omitted Calculations\label{sec:Omitted-Calculations}}

\subsubsection*{Calculating all the stable matchings for Example \ref{exa:multiple-matchings}}

Observe that for any fixed matching statistics $s$, there are unique
market-clearing cutoffs, because both colleges have identical preferences
over students. We have that 
\begin{align*}
u^{\theta}\left(c_{1};s\right)-u^{\theta}\left(c_{2};s\right) & =g_{c_{1}}^{\theta}+16\chi^{\theta}s_{c_{1}}-g_{c_{2}}^{\theta}-16\chi^{\theta}s_{c_{2}}\\
 & =16\chi^{\theta}\left(s_{c_{1}}-s_{c_{2}}\right)-1\,.
\end{align*}
If $s_{c_{2}}>s_{c_{1}}$ all students $\theta$ have preferences
$c_{2}\succ^{\theta|s}c_{1}$. The unique market clearing cutoffs
are $P^{\circ}=\left(0,1/2\right)$, yielding the stable matching
\[
\mu^{\circ}\left(\theta\right)=\begin{cases}
c_{2} & \chi^{\theta}\geq1/2\\
c_{1} & \chi^{\theta}<1/2
\end{cases}
\]
 and the matching statistics $s^{\circ}=\left(1/8,3/8\right)$. 

If $s_{c_{1}}>s_{c_{2}}$ a student $\theta$ has preferences $c_{1}\succ^{\theta|s}c_{2}$
if $\chi^{\theta}>a$ for $a=\frac{1}{16\left(s_{c_{1}}-s_{c_{2}}\right)}$.
If $a<1-q_{1}=1/2$ then the unique market clearing cutoffs are $P^{\dagger}=\left(1/2,0\right)$
yielding the matching 
\[
\mu^{\dagger}\left(\theta\right)=\begin{cases}
c_{1} & \chi^{\theta}\geq1/2\\
c_{2} & \chi^{\theta}<1/2
\end{cases}
\]
 and the matching statistics $s^{\dagger}=\left(3/8,1/8\right)$.
To verify $\mu^{\dagger}$ is a stable matching, calculate that given $s^{\dagger}$
we have that $a^{\dagger}=\frac{1}{16\left(s_{c_{1}}^{\dagger}-s_{c_{2}}^{\dagger}\right)}=1/4<1/2$. 

If $1/2\leq a\leq1$ the unique market clearing cutoffs are $P^{a}=\left(0,a-1/2\right)$
yielding the matching 
\[
\mu^{a}\left(\theta\right)=\begin{cases}
c_{1} & \chi^{\theta}\geq a\text{ or }\chi^{\theta}<a-1/2\\
c_{2} & a-1/2\leq\chi^{\theta}<a
\end{cases}
\]
 because students in $[a,1]$ prefer $c_{1}$, students in $\left[a-1/2,a\right]$
prefer $c_{2}$, and students in $\left[0,a-1/2\right]$ have a budget
set equal to $\left\{ c_{1}\right\} $. The corresponding matching
statistics are $s^{a}=\sigma\left(\mu^{a}\right)=\left(\frac{5-4a}{8},\frac{4a-1}{8}\right)$.
We have that $a^{\vartriangle}=\left(3+\sqrt{5}\right)/8\approx0.65$
is the unique $1/2\leq a\leq1$ such that 
\[
a=\frac{1}{16\left(s_{c_{1}}^{a}-s_{c_{2}}^{a}\right)}=\frac{1}{12-16a}\;.
\]
Therefore, the last stable matching is 
\[
\mu^{\vartriangle}\left(\theta\right)=\begin{cases}
c_{1} & \chi^{\theta}\geq a^{\vartriangle}\text{ or }\chi^{\theta}<a^{\vartriangle}-1/2\\
c_{2} & a^{\vartriangle}-1/2\leq\chi^{\theta}<a^{\vartriangle}
\end{cases}
\]
and the corresponding matching statistics are $s^{\vartriangle}=\left(\frac{5}{8}-\frac{1}{2}a^{\vartriangle},\frac{1}{2}a^{\vartriangle}-\frac{1}{8}\right)=\left(\frac{7-\sqrt{5}}{16},\frac{1+\sqrt{5}}{16}\right)\approx\left(0.3,0.2\right)$.
\end{document}